\documentclass[11pt]{article}

\usepackage{natbib}
\usepackage{booktabs} 
\usepackage[ruled,linesnumbered]{algorithm2e} 

\SetAlFnt{\small}
\SetAlCapFnt{\small}
\SetAlCapNameFnt{\small}
\SetAlCapHSkip{0pt}
\IncMargin{-\parindent}

\renewcommand{\cite}{\citep}

\usepackage{fullpage}

\usepackage{amsmath,amsfonts,mathtools,amsthm,amssymb,thm-restate}

\usepackage{hyperref} 
\hypersetup{
    colorlinks=true,
    allcolors=teal,
    pdfencoding=auto,
}

\usepackage{cleveref}
\usepackage{multirow}
\usepackage{tikz}
\usepackage{tikz-dependency}
\usepackage{enumitem}

\usepackage{tabularx}
\usepackage{tikz}

\newcolumntype{M}{>{\centering\arraybackslash}m{1cm}}
\newcommand\tikzmark[2]{%
\tikz[remember picture,baseline] \node[inner sep=0.1pt,outer sep=0] (#1){#2};%
}

\newcommand\link[2]{%
\begin{tikzpicture}[remember picture, overlay, >=stealth, shift={(0,0)}]
  \draw[-implies,double equal sign distance] (#1) to (#2);
\end{tikzpicture}%
}
\usepackage{capt-of}
\usepackage{makecell}
\newcommand{\edit}[1]{{#1}}

\newtheorem{theorem}{Theorem}
\newtheorem{definition}{Definition}
\newtheorem{corollary}{Corollary}
\newtheorem{lemma}{Lemma}

\newtheorem{proposition}{Proposition}
\newtheorem{claim}{Claim}

\DeclareMathOperator*{\E}{\mathbb{E}}

\newcommand{\set}[1]{\left\{#1\right\}}
\newcommand{\norm}[1]{\left\lVert#1\right\rVert}

\usepackage{color}

\newcommand{\ENVY}{\mathsf{Envy}}
\newcommand{\env}{\ENVY}
\newcommand{\envy}{\ENVY}

\newcommand{\OPT}{\mathsf{OPT}}
\newcommand{\defeq}{\vcentcolon=}

\newcommand{\dist}{\mathcal{D}}

\newcommand{\agents}{\mathcal{N}}

\newcommand{\Sph}{\mathcal{S}}
\newcommand{\Ball}{\mathcal{B}}

\newcommand{\calC}{\mathcal{C}}
\newcommand{\calT}{\mathcal{T}}
\newcommand{\calD}{\mathcal{D}}
\newcommand{\calW}{\mathcal{W}}
\newcommand{\calS}{\mathcal{S}}
\newcommand{\calA}{\mathcal{A}}
\newcommand{\conv}{\mathrm{Conv}}
\DeclareMathOperator*{\argmin}{arg\,min}
\DeclareMathOperator*{\argmax}{arg\,max}

\DeclarePairedDelimiter{\ceil}{\lceil}{\rceil}
\DeclarePairedDelimiter{\floor}{\lfloor}{\rfloor}
\allowdisplaybreaks

\newcommand\numberthis{\addtocounter{equation}{1}\tag{\theequation}}

\title{Online Envy Minimization and Multicolor Discrepancy: Equivalences and Separations}

\author{
Daniel Halpern\thanks{Harvard University. Email: dhalpern@g.harvard.edu} \and 
Alexandros Psomas\thanks{Purdue University. Email: apsomas@purdue.edu} \and 
Paritosh Verma\thanks{Purdue University. Email: verma136@purdue.edu} \and 
Daniel Xie\thanks{GEICO. Email: danielyxie2002@gmail.com. Work conducted while the author was at Purdue University.}
}
\date{}

\begin{document}

\maketitle
\thispagestyle{empty}

\begin{abstract}
We consider the fundamental problem of allocating $T$ indivisible items that arrive over time to $n$ agents with additive preferences, with the goal of minimizing \emph{envy}. This problem is tightly connected to the problem of \emph{online multicolor discrepancy}: vectors $v_1, \dots, v_T \in \mathbb{R}^d$ with $\norm{v_i}_2 \leq 1$ arrive one at a time and must be, immediately and irrevocably, assigned to one of $n$ colors to minimize $\max_{i,j \in [n]} \norm{ \sum_{v \in S_i} v - \sum_{v \in S_j} v }_{\infty}$ at each step, where $S_\ell$ is the set of vectors that are assigned color $\ell$. The special case of $n = 2$ is called \emph{online vector balancing}, introduced by Spencer nearly half a century ago ~\cite{spencer1977balancing}. It is known that multicolor discrepancy is at least as hard as envy minimization: any bound for the former implies the same bound for the latter. Against an adaptive adversary, both problems have the same optimal bound: $\Theta(\sqrt{T})$; it is not known, however, whether the optimal bounds match against weaker adversaries.

Against an oblivious adversary, \citet{alweiss2021discrepancy} give an elegant upper bound of $O(\log T)$, with high probability, for the online multicolor discrepancy problem. In a recent breakthrough, \citet{kulkarni2024optimal} improve this to $O(\sqrt{\log T})$ for the case of online vector balancing and give a matching lower bound. However, it has remained an open problem whether a $O(\sqrt{\log T})$ bound is possible for multicolor discrepancy. Furthermore, these results give, as corollaries, the state-of-the-art upper bounds for online envy minimization (against an oblivious adversary) for $n$ and two agents, respectively; it is an open problem whether better bounds are possible.

In this paper, we resolve all aforementioned open problems. We establish that online envy minimization is, in fact, equivalent to online multicolor discrepancy against the oblivious adversary: we give an upper bound of $O(\sqrt{\log T})$, with high probability, for multicolor discrepancy, and a lower bound of $\Omega(\sqrt{\log T})$ for envy minimization, resolving both problems. We proceed to study weaker adversaries, where we prove that the two problems are no longer equivalent. Against an i.i.d. adversary, we establish a separation: For online vector balancing, we give a lower bound of $\Omega\left(\sqrt{\frac{\log T}{\log \log T}}\right)$, while for envy minimization, we give an algorithm that guarantees a constant upper bound.
\end{abstract}

\section{Introduction}

We consider the fundamental problem of fairly allocating indivisible items that arrive sequentially over time to agents with additive preferences. At each time step $t$, an item $g_t$ arrives and must be allocated immediately and irrevocably to one of $n$ agents. Each agent $i$ has value $v_{i,t} \in [0,1]$ for $g_t$, which is revealed only upon the item's arrival. Our objective is to ensure that the final allocation is fair, which we measure through the notion of \emph{envy-freeness}. Concretely, our goal is to minimize $\max_{i,j \in [n]} \left( \sum_{t \in A_j^T} v_{i,t} - \sum_{t \in A_i^T} v_{i,t} \right)$, where $A_{\ell}^T$ is the allocation of agent $\ell$ after all $T$ items have been allocated. 
The canonical motivation for this problem is that of food banks~\cite{aleksandrov2015online,lee2019webuildai,mertzanidis2024}, which receive food donations and allocate them to non-profits (e.g., food pantries). Ensuring that no organization is significantly disadvantaged relative to others is a natural challenge in such settings.

A crucial question in this problem is what assumptions can we reasonably make about the items' values? Different assumptions lead to vastly different fairness guarantees.

At one extreme, many works assume that the vector of values for $g_t$, $(v_{1, t}, \ldots, v_{n, t})$, is drawn independently from a fixed distribution $\mathcal{D}$. This assumption allows for strong positive results: If $T$ is sufficiently large, with high probability, we can find allocations that are completely envy-free, achieving a maximum of at most $0$. However, this setting has notable downsides. First, in many practical scenarios, $T$ may need to be quite large before these guarantees become meaningful. Furthermore, the distributional assumption rules out natural hard instances---such as a single high-value item that agents agree is better than all others combined---where envy-freeness is fundamentally unattainable.

At the other extreme, we can consider fully adversarial settings, where item values are chosen \emph{adaptively}: an adversary, observing the allocation algorithm and all allocation decisions made so far, selects each item's values in order to maximize envy at the end. Prior work has shown that in such a setting, any online algorithm must incur envy at least $\Omega(\sqrt{T})$, even when there are only $n=2$ agents~\cite{benade2024fair}. While this model captures the absolute worst-case scenario, it may be overly pessimistic for many real-world applications. 

In this paper, we explore intermediate adversaries. One particularly natural, yet robust, assumption is that of an \emph{oblivious adversary}, which selects worst-case item values in advance, knowing only the algorithm, and not depending on the specific allocation choices made in previous rounds. To better understand this problem, it is useful to consider the closely related \emph{online vector balancing} problem, introduced by Spencer nearly half a century ago~\cite{spencer1977balancing}, and its generalization, \emph{online multicolor discrepancy}.

In online vector balancing, at each time step $t$, a vector $v_t$ with $\norm{v_t}_2 \leq 1$ arrives and must be assigned immediately and irrevocably to one of two bundles. The key object of study is the \emph{discrepancy vector at time $t$}: $d_t := \sum_{v \in S_1^t} v - \sum_{v \in S_2^t} v$ where $S_1^t$ and $S_2^t$ be the set of vectors assigned to each bundle after $t$ steps. The goal is to keep the $\ell_{\infty}$-norm of all $d_t$s as small as possible.\footnote{The classic problem is more frequently formalized as choosing a sign $\chi_t \in \{-1, 1\}$, and setting $d_t := \sum_{i = 1}^t \chi_i v_i$. Our formulation is equivalent, using notation more consistent with envy-minimization.} That is, find the smallest $B$ for which $\max_{t \in [T]} \norm{ d_t }_{\infty} \leq B$ (or, if randomness is involved, given a $\delta$, the smallest $B$ for which this holds with probability $1 - \delta$).

In online multicolor discrepancy, instead of \emph{two} bundles, there are $n$ bundles~\cite{bansal2021online}. At each time step $t$, a vector $v_t \in \mathbb{R}^d$ with $\norm{v_t}_2 \le 1$ arrives and must be assigned one of $n$ bundles (often framed as assigning the vector one of $n$ colors). The goal is to minimize the maximum discrepancy between any pair of bundles, defined as $\max_{t \in [T], i,j \in [n]} \norm{ \sum_{z \in S^t_{i}} z - \sum_{z \in S^t_{j}} z }_{\infty}$ where $S_i^t$ denotes the set of vectors assigned to bundle $i$ after $t$ steps. Online vector balancing corresponds to the special case of $n = 2$.

These problems are particularly useful because algorithms for online multicolor discrepancy can be directly applied to envy-minimization: the item values $(v_{1, t}, \ldots, v_{n, t})$ can be treated as input vectors to a discrepancy algorithm, and the resulting envy is upper bounded by the discrepancy.\footnote{Since envy-minimization values are only bounded in $[0, 1]$, the $\ell_2$-norm may be as large as $\sqrt{n}$, increasing the bounds by a factor of $\sqrt{n}$. However, since we primarily focus on bounds as a function of $T$, this detail is less critical.}  Additionally, algorithms for online vector balancing can be used for envy-minimization when $n = 2$.
However, discrepancy problems are generally more challenging than envy-minimization because of the following reasons: 
(i) input vectors in discrepancy problems may have negative entries, whereas item values in envy minimization are nonnegative, (ii) discrepancy requires bounding the $\ell_{\infty}$ norm, whereas in envy minimization, negative envy of large magnitude is desirable,\footnote{In fact, a consequence of this is that positive results for online multicolor discrepancy give online algorithms for computing a near-perfect allocation (a perfect allocation $A$ satisfies $v_i(A_j) = v_{i}([m])/n$ for all agents $i,j$), a problem harder than computing a small-envy allocation.} and (iii) in envy-minimization the vectors always have dimension $d=n$, while in discrepancy we need to handle arbitrary combinations of $n$ and $d$.

We now summarize the best-known bounds for these problems.
For online vector balancing, a simple greedy algorithm achieves a bound of $O(\sqrt{T})$ against an adaptive adversary, which was shown to be tight by \citet{spencer1977balancing,spencer1994ten}. Against a weaker, oblivious adversary, \citet{alweiss2021discrepancy} proposed an elegant algorithm that guarantees an $O(\log T)$ bound with high probability. In a recent breakthrough, \citet{kulkarni2024optimal} give an algorithm that guarantees an $O(\sqrt{\log T})$ bound with high probability, as well as a matching lower bound, thus resolving the optimal dependence on $T$ for the vector balancing problem against an oblivious adversary.
An immediate implication of this result is a $O(\sqrt{\log T})$ bound, with high probability, for online envy minimization with $n=2$ agents against an oblivious adversary; however, no corresponding lower bound for envy minimization was known.\footnote{It is worth mentioning that \cite{benade2024fair} prove that sublinear envy is incompatible with non-trivial efficiency guarantees against an oblivious adversary. This result, however, has no implications for envy minimization.} For online multicolor discrepancy against an adaptive adversary, a $O(\sqrt{T})$ bound is possible, while against an oblivious adversary, the best known result was an $O(\log T)$ bound with high probability~\cite{alweiss2021discrepancy}. This directly implies an $O(\log T)$ with high probability bound for online envy minimization problem with $n$ agents.

Overall, the state-of-the-art can be summarized as follows. Online envy minimization is as hard as online multicolor discrepancy against an adaptive adversary, i.e., a tight $\Theta(\sqrt{T})$ bound is known for both problems. And, for all we know, this could be the case for an oblivious adversary as well: there are no known lower bounds for online envy minimization, and the best known upper bounds for $n=2$ and $n>2$ agents are implications of online vector balancing and online multicolor discrepancy, respectively. There is a gap between the best-known bound for online vector balancing ($O(\sqrt{\log T})$, which is known to be optimal) and online multicolor discrepancy ($O(\log T)$, not known to be optimal).
Finally, for all we know, online envy minimization is as hard as online multicolor discrepancy for stochastic adversaries, weaker than an oblivious adversary, where we only know lower bounds for online multicolor discrepancy.\footnote{As we explain in detail later in the paper, known positive results for online envy minimization against stochastic adversaries, e.g.~\cite{benade2024fair}, rely on certain technical assumptions, rendering them incompatible with the online vector balancing literature.}
Simply put, the goal of this paper is to resolve these gaps in our understanding of online envy minimization and online multicolor discrepancy.

\subsection{Our results}

In~\Cref{thm: multi color main upper bound} we prove the existence of an algorithm for the online multicolor discrepancy problem that achieves a bound of $O(\sqrt{\log T})$, with high probability, against an oblivious adversary. This result directly implies the same bound for online envy minimization with $n$ agents (\Cref{cor: main result for oblivious and fair division}). We also give a matching lower bound in \Cref{thm: envy lower bound for oblivious}: for any $r \in (0,1)$ and $c > 0$, an oblivious adversary can guarantee envy at least $\Omega((\log(T))^{r/2})$ with probability at least $1/T^c$. Thus, we obtain optimal bounds for both problems, online multicolor discrepancy and envy minimization, against an oblivious adversary.\footnote{Our lower bound rules out the possibility of achieving $o(\sqrt{\log{T}})$ envy with probability at least $1-1/\mathrm{poly}(T)$. However, achieving a $o(\sqrt{\log{T}})$ envy with a constant probability (independent of $T$) may be possible.} Overall, these results show that similar to the case of an adaptive adversary, online envy minimization is no easier than online multicolor discrepancy against an oblivious adversary.

Next, we analyze a stochastic, i.i.d.\@ adversary that selects a distribution $\calD$ such that, in online multicolor discrepancy (respectively, envy minimization), $v_{t,i} \sim \calD$, for all rounds $t$ and $i \in [d]$ (respectively, $v_{i,t} \sim \calD$, for all agents $i$ and items $g_t$). In the discrepancy minimization literature, \citet{bansal2020line} consider a similar model, where each coordinate of the arriving vectors is sampled uniformly at random from the set $\{ -1, 1\}$, achieving an $O(\sqrt{n} \log T)$ bound with high probability. Other works in the discrepancy minimization world use the term ``i.i.d.'' to refer to settings where dependence over the coordinates is allowed; e.g., in~\cite{bansal2020online} vectors $v_1, \dots v_T$ are sampled i.i.d.\@ from a distribution over $[-1,1]^d$; here, a lower bound of $\Omega\left( \sqrt{\frac{\log T}{\log \log T}} \right)$ is known~\cite{bansal2020online}. In~\Cref{thm: lower bound for iid vector balancing} we prove that there exists a distribution for which every algorithm must have discrepancy at least $\Omega\left( \sqrt{\frac{\log T}{\log \log T}} \right)$ with high probability, even in our ``easier'' i.i.d.\@ model (where coordinates of each vector are also sampled i.i.d.). This result implies that the upper bound in~\Cref{thm: multi color main upper bound} against an oblivious adversary cannot be improved, up to $\log\log$ factors, even against a much weaker i.i.d.\@ adversary.

In the fair division literature, many works provide guarantees in stochastic models. However, to the best of our knowledge, all previous results are asymptotic with respect to the number of items. Specifically, as observed by \citet{bansal2020online}, an innocuous-looking (and prevalent) assumption in stochastic fair division is that the adversary's distribution does not depend on the number of items $T$ (e.g., ruling out a variance of $1/T$). The setup is typically as follows: given a number of agents $n$, the adversary specifies a distribution $\dist$. The designer, who knows this distribution, then selects a (possibly randomized) algorithm. Nature then selects a number of items $T$. In every round $t$, the value of item $t$ for each agent $i$ is drawn independently from $\dist$, i.e., $v_{i,t} \sim \dist$. Under such a ``constant distribution'' assumption, envy-freeness with ``high probability'' means ``probability at least $1 - O(1/\text{poly}(T))$,'' where $\dist$ is treated as a constant. In this easier setup, envy-freeness with high probability is known to be compatible with Pareto Efficiency ex-post, even online~\cite{benade2024fair}. Removing the ``constant distribution'' assumption introduces numerous technical obstacles; see~\Cref{subsec: why random is hard} for a discussion. In this paper, we give an online algorithm that guarantees envy of at most $c+1$ with probability at least $1 - O(T^{-c/2})$ regardless of the distribution.

See~\Cref{tab:summary_results} for a summary of our results.

\begin{table}[t]
    \centering
    \caption{Our results for Online Multicolor Discrepancy (OMD) and Envy Minimization (OEM)}\label{tab:summary_results}
    \resizebox{\linewidth}{!}{%
    \begin{tabular}{p{1.1cm}cc|cc}
    \toprule
    \multirow{2}{*}{} & \multicolumn{2}{c|}{\textbf{Oblivious}} & \multicolumn{2}{c}{\textbf{i.i.d.}} \\
    \cline{2-5}
    & \textbf{Upper Bound} & \textbf{Lower Bound} & \textbf{Upper Bound} & \textbf{Lower Bound} \\
    \hline
    \tikzmark{z1}{OMD}
    & \tikzmark{a1}{$O(\sqrt{\log T})$ [Thm~\ref{thm: multi color main upper bound}]}
    & \tikzmark{b1}{$\Omega((\log T)^{r/2})$} 
    & \tikzmark{a2}{$O(\sqrt{\log T})$ [Thm~\ref{thm: multi color main upper bound}]} 
    & \tikzmark{b2}{$\Omega\left(\sqrt{\frac{\log T}{\log \log T}}\right)$ [Thm~\ref{thm: lower bound for iid vector balancing}]} \\
    \hline
    \tikzmark{z2}{OEM} 
    & \tikzmark{a3}{$O(\sqrt{\log T})$ [Cor.~\ref{cor: main result for oblivious and fair division}]} 
    & \tikzmark{b3}{$\Omega((\log T)^{r/2})$ [Thm.~\ref{thm: envy lower bound for oblivious}]} 
    & \tikzmark{a4}{$O(1)$ [Thm.~\ref{thm:n agent upper bound iid fair division}]} 
    &  - \\
    \hline
    \end{tabular}%
    } 
    \link{a1}{a3}
    \link{b3}{b1}
\end{table}

\subsection{Technical overview: multicolor discrepancy against an oblivious adversary}\label{subsec: technical overview of multicolor}


To obtain an optimal algorithm for online vector balancing, \citet{kulkarni2024optimal} equivalently view this problem in terms of a (massive) rooted tree $\calT = (V, E)$ having a depth $T$. Each edge $e$ of $\calT$ that lies at depth $t$ corresponds to a possible choice of vector $v_e$ that the adversary can pick (after some discretization of the adversary's available options) at time step $t$. An oblivious adversary picks a path from the root to a leaf, which is revealed one edge at a time, and as an edge $e$ appears, the algorithm must assign a (possibly random) sign $x_e \in \{-1,1 \}$ to it. The first step in \citeauthor{kulkarni2024optimal}'s proof is to show that, for any convex body $K$ that is sufficiently large, and any rooted tree $\calT = (V,E)$ (whose edges have associated vectors), there exist signs $x \in \{-1,1 \}^E$ such that, for some constant $\alpha < 5$, $\sum_{e \in P} x_e v_e \in \alpha K$, for any root-to-node path $P$. This result generalizes a well-known result of \citet{banaszczyk2012series} from paths to trees. Then, using a carefully designed large convex body $K$, \citeauthor{kulkarni2024optimal} show how to go from a single choice of signs $x$ to a distribution $\calD$ over signs, so that for $y \sim \calD$, $\sum_{e \in P} y_e v_e$ is $\gamma$-subgaussian for all root to leaf paths $P$, for some constant $\gamma$. Finally, they show that by taking a fine enough $\varepsilon$-net of the unit ball and constructing an appropriate rooted tree $\calT$, whose edges correspond to vectors of the $\varepsilon$-net, the distribution $\calD$ over signs implies an optimal algorithm for online vector balancing.

Our approach for online multicolor discrepancy follows similar high-level steps. Intuitively, in \citet{kulkarni2024optimal}'s blueprint for online vector balancing, we could think of the adversary as picking a set $S_e = \{ v_e, - v_e \}$, and the algorithm picking one of the two vectors from $S_e$ for each edge $e$. The most challenging step of our proof is to extend the result of~\cite{kulkarni2024optimal} from trees where edges have associated vectors (or sets of size two) to trees where edges have associated sets of arbitrary size. Concretely, in \Cref{theorem:tree-reduction} we prove that given any rooted tree $\calT = (V,E)$ where every edge $e \in E$ has an associated set of vectors $S_e \subseteq \Ball^d_2$ (satisfying a couple of technical assumptions, for example, $\mathbf{0} \in \conv(S_e)$), and a sufficiently large symmetric convex body $K$, there exists a vector $v_e \in S_e$, such that for all $u \in V$, $\sum_{e \in P_u} v_e \in 11 K$, where $P_u$ is the set of edges of the path from the root to the node $u$. To prove \Cref{theorem:tree-reduction} we start with a fractional selection $x^e = (x^e_1, \dots, x^e_{|S_e|})$ of vectors for each edge $e$, such that the desired property is satisfied. We iteratively round this fractional selection, so that each step does not increase $\sum_{e \in P_u} x^e v_e$ by too much, so that the final, integral selection has the desired property. Our rounding is \emph{bit-by-bit}, inspired by similar procedures by~\cite{bansal2022prefix,lovasz1986discrepancy}. Our process has two steps. First, we prove that each $x^e_i$ can be rounded so that its fractional part is at most $k$ bits long (for some small $k$). Then, we iteratively reduce the number of bits in the fractional part by rounding each remaining bit, one at a time, until all values become integers. The rounding decisions---which bits get rounded down to $0$ and which bits are rounded up to $1$---are guided by a black-box call to~\cite{kulkarni2024optimal}'s extension of Banaszczyk's result to a certain blown-up tree (where the $\{-1,1\}$ signs tell us whether to round up or down).

With~\Cref{theorem:tree-reduction} in hand, we prove there exists a distribution $\calD$ over vectors (one from each edge set) such that for $x \sim \calD$, $\sum_{e \in P_u} x_e$ is subgaussian, for every node $u \in V$ (\Cref{theorem:tree-subgaussianity}). Finally, we prove that there exists an algorithm that, given sets of vectors one at a time, selects a vector from each set such that the vector sum is $O(1)$-subgaussian (\Cref{theorem:subgauss-algo}). This algorithm can be used to get an algorithm for \emph{weighted} online vector balancing (\Cref{lemma:weighted-vector-balancing}), which, in turn, can be used to give an algorithm for online multicolor discrepancy (\Cref{thm: multi color main upper bound}).

\subsection{Technical overview: online item allocation against a stochastic adversary}\label{subsec: why random is hard}

For the case of an i.i.d.\@ adversary, where each $v_{i,t} \sim \mathcal{D}$, we establish even tighter bounds: the maximum envy is a \emph{constant} with all but polynomially small probability in $T$. Strikingly, the algorithms achieving this guarantee are distribution-independent.

Prior work~\cite{dickerson2014computational} shows that, for a fixed distribution $\calD$, assuming $T$ is sufficiently large, the simple \emph{welfare maximization} algorithm---``allocate each item $t$ to $\argmax_{i \in [n]} v_{i,t}$''---achieves envy-freeness with high probability. Here, we seek bounds that depend only on $T$, and are independent of $\calD$. Unfortunately, for \emph{any} choice of $T$, there exist distributions $\mathcal{D}$ where welfare maximization results in $\Theta(\sqrt{T})$ envy.

Our algorithm works in two phases. The first phase is welfare maximization; this phase might generate significant envy. The second phase, consisting of the final $\Theta(\log T \sqrt{T})$ items, mitigates this envy. At every step $t$, we single out a set of agents who have not received a large
number of items (within this phase). Among this set, we allocate item $t$ to the agent who is envied the least by agents in this set. The key challenges are (i) ensuring that phase two is sufficiently long so that the envy generated during phase one is eliminated, and (ii) preventing endless cycles of envy, where eliminating the envy of agent $i$ inadvertently causes another agent $j$ to envy $i$. As we show, running welfare maximization for longer, i.e., having a longer phase one, makes envy cycles lighter (but increases the maximum envy). Therefore, properties (i) and (ii) are, in fact, in tension: the length of the two phases must be chosen to simultaneously satisfy these competing requirements.

Let $H^t$ be a graph with agents as nodes and an edge $(i,j)$ if, at step $t$, agent $i$ envies agent $j$ by at least a constant $c$. Our goal is to ensure that $H^t$ is empty after all items have been allocated. Intuitively, allocating an item to an agent $i$ decreases the prevalence of outgoing edges and increases the prevalence of incoming edges of node $i$. Moreover, after running welfare maximization phase, $H^t$ is guaranteed to be acyclic. To prevent new edges from forming, during phase 2, we allocate arriving items to sources in $H^t$.

Our first major technical hurdle (\Cref{lem:no-cycle}) is ensuring that no cycles form in $H^t$, for any $1 \leq t \leq T$, which would make allocating to source nodes impossible (and seemingly make eliminating envy extremely challenging). The second major hurdle (\Cref{lem:high-value-n}) is showing that giving agent $i$ a moderate number of items more than agent $j$ during phase 2---specifically $\lceil \log T \sqrt{T} \rceil$---suffices to remove the edge $(i, j)$. Crucially, both of these are statements about \emph{all} time steps $t$. However, phase 2 induces correlations between $H^t$ and $H^{t+1}$, making it difficult to maintain an analytical handle on this graph.

To address this, we introduce an alternate sampling method that, from the algorithm's perspective, is identical to the standard sampling process. We pre-sample a large pool of items (in quantile space) and choose which item arrives next based on which agent will receive it. This is only possible because phase 2 decisions are agnostic to the arriving item's values. Importantly, this approach makes $H^t$ depend only on the \emph{number} of times each agent received an item rather than the specific time steps at which they were received, giving traction toward the lemmas. 

A final challenge is ensuring that our results hold for \emph{any} distribution without additional assumptions. Crucially, we cannot apply standard concentration inequalities, as those typically require a large number of items or specific distributional properties (e.g., a lower bound on variance). To address this, we derive a new distribution-agnostic concentration inequality (tailored to our specific task of bounding the envy), which might be of independent interest.


\subsection{Related work}

\paragraph{Online vector balancing.}
    When the incoming vectors satisfy $\norm{v_i}_2 \leq 1$, then randomly coloring the vectors $\{-1,+1\}$ achieves a discrepancy of $\sqrt{T \log{d}}$, even for the case of an adaptive adversary. A  matching lower bound of $\Omega(\sqrt{T})$ was proved by \citet{spencer1977balancing,spencer1994ten}. The stochastic setting for online vector balancing was first studied in~\cite{bansal2020line}. They considered a setting where the incoming vectors are chosen i.i.d.\@ from the uniform distribution over all vectors in $\{-1,1\}^n$ and gave an algorithm for achieving a $O(\sqrt{d}\log{T})$ bound on discrepancy at all time steps till $T$. This was later improved by \citet{bansal2020online} who showed a $O(d^2 \log{nT})$ upper bound for any distribution supported on $[-1,1]^n$. The work of \citet{bansal2021online} improved the dependence on $d$ by showing that $O(\sqrt{d}\log^4{dT})$ discrepancy can be achieved. The dependence on $T$ was further improved to $O_d(\sqrt{\log(T)})$ by \citet{aru2016balancing}, where the implicit dependence on $d$ was super-exponential. 
    
    The setting of an oblivious adversary remained relatively under-explored until the recent work of \citet{alweiss2021discrepancy}. They design an extremely simple and elegant algorithm that achieves a bound of $O(\log{nT})$ for both online vector balancing and online multicolor discrepancy. Their algorithm is based on a self-balancing random walk, which, for vector balancing, ensures that the discrepancy prefix vector is $O(\sqrt{\log{nT}})$-subgaussian. A tight bound of $\Theta(\sqrt{\log{T}})$ for online vector balancing was achieved by \citet{kulkarni2024optimal}, who proved the existence of an algorithm that maintains $O(1)$-subgaussian prefix vectors.

\paragraph{Stochastic fair division.}

Stochastic fair division, introduced by \citet{dickerson2014computational}, studies the existence of fair allocations when valuations are drawn from a distribution. \citeauthor{dickerson2014computational} show that maximizing utilitarian welfare produces an envy-free allocation with high probability when the number of items $T \in \Omega(n \log n)$ and items values are drawn i.i.d.\@ from a fixed ``constant distribution'' (i.e., the distribution does not depend on the number of items). \citet{manurangsi2020envy,manurangsi2021closing} establish tight bounds for the existence of envy-free allocations in the ``constant distribution'' i.i.d.\@ model: $T \in \Omega (n \log n / \log \log n )$ is both a necessary and sufficient condition. \citet{bai2021envy} extend the result to the case of independent but non-identical additive agents.

Beyond envy-freeness, weaker fairness notions such as maximin share fair~\cite{kurokawa2016can,amanatidis2017approximation,farhadi2019fair} and proportional~\cite{suksompong2016asymptotic} allocations exist with high probability. Finally, the existence of fair allocations for agents with non-additive stochastic valuations is studied in~\cite{manurangsi2021closing,gan2019envy,benade2024existence}.


\paragraph{Online fair division.}
A rich literature studies online or dynamic fair division. Numerous works study the problem where divisible or indivisible items arrive over time, with the goal of optimizing various objectives, including utilitarian welfare~\cite{gkatzelis2021fair,bogomolnaia2022fair}, egalitarian welfare~\cite{springer2022online,kawase2022online}, Nash welfare~\cite{gao2021online,banerjee2022online,liao2022nonstationary,huang2023online,yang2024online}, the generalized mean of the agents' utilities~\cite{barman2022universal}, and approximation of the maximin share guarantee~\cite{zhou2023multi}.

Closer to our work, the aforementioned works~\cite{dickerson2014computational,bai2021envy} prove that maximizing welfare and weighted welfare---algorithms that can be implemented online---achieve envy-freeness with high probability, in addition to Pareto efficiency, when valuations are drawn i.i.d.\@ from ``constant'' distributions for identical and non-identical agents, respectively. \citet{benade2024fair} prove that even when correlation is allowed between agents (but items are independent), Pareto efficiency and fairness are compatible in the online setting---the specific fairness guarantee for a pair of agents $i,j$ is ``either $i$ envies $j$ by at most 1 item, or $i$ does not envy $j$ with high probability''; again, distributions are treated as constants. To the best of our knowledge, we are the first to explore the stochastic setting (online or offline) without the ``constant distribution'' assumption.

There are several variations of the standard models, e.g., revocable allocations~\cite{he2019achieving,yang2023fairly}, only having access to pairwise comparisons~\cite{benade2025dynamic}, two-sided matching~\cite{mertzanidis2024}, items arriving in batches~\cite{benade2018make}, and repeated allocations~\cite{igarashi2024repeated}. Further afield, many works study the ``inverse'' problem of allocating static resources to agents that arrive/depart over time~\cite{kash2014no,li2018dynamic,sinclair2022sequential,vardi2022dynamic,banerjee2023online,kulkarni2025online}.

\section{Preliminaries}\label{sec:prelims}

Throughout the paper we will use $\Sph^{d-1} \defeq \{ x \in \mathbb{R}^d : \norm{x}_2 = 1 \}$ to denote the Euclidean sphere in $\mathbb{R}^d$ and $\Ball^d_2 \vcentcolon= \{ x \in \mathbb{R}^d : \norm{x}_2 \leq 1 \}$ for the Euclidean ball in $\mathbb{R}^d$. 

\paragraph{Online multicolor discrepancy.} 

In the online multicolor discrepancy problem, there is a set of $T$ vectors, $v_1, v_2, \dots, v_T \in \Ball^d_2$. At each time step $t \in [T]$, vector $v_t$ arrives and must be immediately and irrevocably assigned to one of $n$ colors. Our algorithms learn $v_t$ at step $t$. Let $S^t_i$ be the set of vectors that are assigned color $i$ up until time $t$. The \emph{discrepancy} at time $t$ is defined as $\max_{i,j \in [n]} \norm{ \sum_{v \in S^t_i} v - \sum_{v \in S^t_j} }_{\infty}$.
Our goal in the online multicolor discrepancy problem is to minimize the maximum discrepancy over all time steps, that is, minimize $$\max_{t \in [T]} \max_{i,j \in [n]} \norm{ \sum_{v \in S^t_i} v - \sum_{v \in S^t_j} v }_{\infty}.$$ The online vector balancing problem is the special case of the multicolor discrepancy problem where $n=2$.

\paragraph{Online envy minimization.} 

In the online envy minimization problem, there is a set of $T$ indivisible items (also referred to as goods) to be allocated to a set of $n$ agents $\agents$. At each time step $t \in [T]$, item $g_t$ arrives and must be immediately and irrevocably allocated to one of the agents in $\agents$. Item $g_t$ has an associated vector $v_t \in [0,1]^n$ such that $v_t = (v_{1,t}, v_{2,t}, \ldots, v_{n,t})$ where $v_{i,t}$ denotes agent $i$'s value for item $g_t$. Our algorithms learn $v_t$ at step $t$. Let $A^t = (A^t_1, A^t_2, \ldots, A^t_n)$ be the allocation at the end of time $t$ (i.e., after $g_t$ has been allocated) such that for each $i \in \agents$, we have $A^t_i \subseteq \{g_1, g_2, \ldots, g_t\}$, $\cup_{i\in \agents} A^t_i = 
\{g_1, g_2, \ldots, g_t\}$, and $A^t_i \cap A^t_j = \emptyset$ for each $i \neq j$. The value of agent $i \in \agents$ for any set of items (also referred as a \emph{bundle}) $S \subseteq \{g_1, g_2, \ldots, g_T\}$ is denoted by $v_i(S) \coloneqq \sum_{g_k \in S} v_{i,k}$, i.e., the preferences of agents are \emph{additive}.

At each timestep $t$, and for any pair of agents $i,j \in \agents$, we define $\envy^t_{i,j}(A^t) \coloneqq v_i(A^t_j) - v_i(A^t_i)$ to be the difference in value, with respect to agent $i$'s preferences, between the bundle of agent $j$ at time $t$ and the bundle of agent $i$ at time $t$. The maximum envy at time $t \leq T$ is denoted by $\envy^t(A^t) = \max_{i,j \in [n]} \ \envy^t_{i,j}(A^t)$. Our goal in the online envy minimization problem is to minimize the maximum envy at time $T$, that is, minimize $\envy^T = \envy^T(A^T)$. 

\paragraph{Adversary models.} 
For both problems, our results crucially depend on how the input (vectors/agents' values) are generated. It will be convenient to think of a game between an adversary, who picks the input, and the designer, who picks an online algorithm.

First, we consider an \textbf{oblivious adversary}.
In the online multicolor discrepancy problem, given a number of colors $n$, a number of vectors $T$, and a dimension $d$, the designer picks a (possibly randomized) algorithm. Then, an oblivious adversary, who knows the algorithm's ``code,'' but does not have access to any of the randomness the algorithm uses, selects all $T$ vectors, which are presented to the algorithm one at a time, in the order the oblivious adversary selected at the start of the game.
In the online envy minimization problem, given a number of agents $n$ and a number of items $T$, the designer picks a (possibly randomized) algorithm. Then, an oblivious adversary, who knows the algorithm's ``code,'' but does not have access to any of the randomness the algorithm uses, selects the agents' values for all $T$ items, which are presented to the algorithm one at a time, in the order the oblivious adversary selected at the start of the game.

Second, we consider an \textbf{i.i.d.\@ adversary}.
In the online multicolor discrepancy problem, given a number of colors $n$, a number of vectors $T$, and a dimension $d$, a stochastic adversary specifies a distribution $\dist$, supported on $[-1,1]$. The designer, who knows this distribution, then selects a (possibly randomized) algorithm. In every round $t$, the vector $v_t$ is generated by sampling each of its coordinates (independently) from $\dist$, i.e., $v_{t,j} \sim \dist$ for all $j \in [d]$. In the online envy minimization problem, given a number of agents $n$ and a number of items $T$, the stochastic adversary specifies a distribution $\dist$, supported on $[0,1]$. The designer, who knows this distribution, then selects a (possibly randomized) algorithm. In every round $t$, the value of item $t$ for each agent $i$ is drawn independently from $\dist$, i.e., $v_{i,t} \sim \dist$.


\subsection{Other technical preliminaries}

\paragraph{Geometry definitions.}

A body $K \subseteq \mathbb{R}^d$ is \emph{convex} if and only if $x, y \in K$ implies that $\alpha x + (1-\alpha) y \in K$ for any $\alpha \in [0,1]$. A body $K\subseteq \mathbb{R}^d$ is \emph{symmetric} if and only if $x \in K$ implies that $-x \in K$ as well. We use $\mathbf{0} \in \mathbb{R}^d$ to denote the origin. The \emph{Gaussian measure} $\gamma_d(K)$ for a body $K \subseteq \mathbb{R}^d$ is defined as $\gamma_d(K) \coloneqq \Pr_{v \sim \mathcal{N}(\mathbf{0}, \mathbf{I}^d)}[v \in K]$ and denotes the probability that a random vector $v$, drawn from the standard Gaussian $\mathcal{N}(\mathbf{0}, \mathbf{I}^d)$ in $\mathbb{R}^d$, is in the body $K$.  We use the following result about convex bodies \cite{ball1997elementary}; the proof is stated below for completeness.

\begin{proposition}\label{proposition:large-body-ball}
    If $\gamma_d(K) \geq 1/2 + \varepsilon$ for a convex body $K \subseteq \mathbb{R}^d$, for some $\varepsilon > 0$, then $\varepsilon \, \Ball^d_2 \subseteq K$.
\end{proposition}
\begin{proof}
    If $K$ does not contain $\varepsilon \, \Ball^d_2$, then there is $p \notin K$ such that $p \in \varepsilon \, \Ball^d_2$; therefore, since $K$ is convex, there must exist a hyperplane $h$ that separates $p$ from $K$, such that $h \cap (\varepsilon \, \Ball^d_2) \neq \emptyset$, and $h$ does not go through the origin. Denote by $H^+ \subseteq \mathbb{R}^d$ the halfspace, created by $h$, that includes the origin, and let $H^{-}$ be the other halfspace. Since $H^-$ is a halfspace that doesn't contain the origin we have that $\gamma_d(H^-) < 1/2$. Using the fact that the standard Gaussian distribution $\mathcal{N}(\mathbf{0}, \mathbf{I}^d)$ is centrally symmetric and that fact that the distance of the origin from $h$ is strictly less than $\varepsilon$, we can bound $\gamma_d(H^+) < \mathbb{P}_{u \sim \mathcal{N}(0,1)}[u \leq \varepsilon] \leq 1/2 + \varepsilon$. This leads to a contradiction, since $\gamma_d(H^+),\gamma_d(H^-) < 1/2 + \varepsilon = \gamma_d(K)$, and $K$ must lie entirely in either $H^+$ or $H^-$.
\end{proof}

Given a set $S \subseteq \mathbb{R}^d$, we will use $\conv(S) \coloneqq \{v : v = \sum_{u \in S} u \cdot x_u \text{, } \sum_{u \in S} x_u = 1 \text{, and } x_u \geq 0 \text{ for all } u \in S\}$ to denote the convex hull of $S$. Given a set $A \subseteq \mathbb{R}^d$, a set $W \subseteq A$ is called an $\varepsilon$-net of $S$ if for all $x \in A$ there is a $y \in W$ such that $\norm{x - y}_2 \leq \varepsilon$. For any $\varepsilon \in (0,1]$, there exists an $\varepsilon$-net of $\Ball^d_2$ of size at most $\left( \frac{3}{\varepsilon} \right)^{d}$; see \cite{kulkarni2024optimal} for a proof. We will use the following analogous statement for $\bigtimes_{i=1}^k \Ball^d_2$.

\begin{proposition}\label{proposition:size-of-net}
    For any $\varepsilon \in (0,1]$, there exists an $\varepsilon$-net of $\bigtimes_{i=1}^k \Ball^d_2$ that is of size at most $\left(\frac{3\sqrt{k}}{\varepsilon}\right)^{dk}$.
\end{proposition}
\begin{proof}
    Let $W$ be an $\varepsilon/\sqrt{k}$-net of $\Ball_2^d$. Since $\Ball^{d}_2$ has an $\varepsilon$-net of size at most $\left(\frac{3}{\varepsilon}\right)^{d}$ (by the result of~\cite{kulkarni2024optimal}), we have $|W| \leq \left(\frac{3\sqrt{k}}{\varepsilon}\right)^{d}$. We will show that $W^k = \bigtimes_{i=1}^k W$ is an $\varepsilon$-net of $\bigtimes_{i=1}^k \Ball^d_2$; this would complete the proof as $|W^k| \leq \left(\frac{3\sqrt{k}}{\varepsilon}\right)^{dk}$. For any $(y_1, y_2, \ldots, y_k) \in \bigtimes_{i=1}^k \Ball^d_2$, let $x_i \in W$ such that $\norm{x_i - y_i}_2 \leq \varepsilon/\sqrt{k}$. If we consider the vector $(x_1, x_2, \ldots, x_k) \in W^k$, then $\sqrt{\sum_{i=1}^k \norm{x_i - y_i}_2^2} \leq \sqrt{k \cdot \varepsilon^2/k} = \varepsilon$, showing that $W^k$ is an $\varepsilon$-net of $W$.
\end{proof}

\paragraph{Subgaussian distributions.}

\begin{definition}[Subgaussian norm]
    For a real-valued random variable $X$, the subgaussian norm is defined as $\norm{X}_{\psi_2} \defeq \inf\{ t > 0: \E[\exp(X^2/t^2)] \leq 2 \}$. For a random vector $Y$ taking values in $\mathbb{R}^d$, $\norm{Y}_{\psi_2, \infty} \defeq \sup_{w \in \Sph^{d-1}} \norm{ \langle\, Y , w \rangle }_{\psi_2}$.
\end{definition}

We say that a random vector $Y$ taking values in $\mathbb{R}^d$ 
is $\beta$-subgaussian if $\norm{Y}_{\psi_2, \infty} \leq \beta$. We use the following two properties of norms: (i) for a random vector $Y$ taking values in $\mathbb{R}^d$ and any $k \geq 0$, we have $\norm{k Y}_{\psi_2, \infty} = k \norm{Y}_{\psi_2, \infty}$ (homogeneity); (ii) for random vectors $Y$ and $Z$ in $\mathbb{R}^d$, we have $\norm{ X+ Y }_{\psi_2, \infty} \leq \norm{X}_{\psi_2, \infty} + \norm{Y}_{\psi_2, \infty}$ (triangle inequality). 





\section{Performance against an oblivious adversary}\label{sec: oblivious}
 
In this section, we prove that there exists an algorithm for the online multicolor discrepancy problem that guarantees, with high probability, a maximum discrepancy of $O(\sqrt{\log T})$ against an oblivious adversary (\Cref{subsec: optimal multicolor discrepancy}); there is a matching lower bound for the online vector balancing problem~\cite{kulkarni2024optimal}, therefore our algorithm for the online multicolor discrepancy problem is optimal. As a corollary, we get an algorithm that guarantees, with high probability, an envy of $O(\sqrt{\log T})$, against an oblivious adversary, for the online envy minimization problem; we give a matching lower bound in~\Cref{subsec: envy for oblivious}. Missing proofs are deferred to~\Cref{app: missing from section 3}.

\subsection{Balancing sets of vectors}\label{subsec: balance sets of vectors}


\citet{kulkarni2024optimal} think of a discretized adversary whose choices correspond to root-to-leaf paths in a (massive) rooted tree. Each edge of this tree corresponds to a vector in $\Ball^d_2$ that the adversary can pick. They prove the following:

\begin{theorem}[\cite{kulkarni2024optimal}]\label{theorem:tree-kulkarni} Let $\calT = (V,E)$ be a rooted tree, where every edge $e \in E$ has a corresponding vector $v_e \in \Ball^d_2$. Let $K \subseteq \mathbb{R}^d$ be a convex body with $\gamma_d(K) \geq 1 - \frac{1}{2|E|}$. Then there exists $z \in \{-1,1\}^{|E|}$ such that, for all $u \in V$, the vector sum $\sum_{e \in P_u} z_e v_e \in 5K$, where $P_u$ is the set of edges of the path from the root to the node $u$.
\end{theorem}

For our result, we think of the adversary as picking sets of vectors in $\Ball^d_2$ at each time step instead of picking individual vectors; notably, the former adversary is strictly more general than the latter. Such an adversary corresponds to a (massive) rooted tree whose each edge $e$ has an associated set of vectors $S_e \subseteq \Ball^d_2$. In the following theorem, we prove that we can always choose a vector from each set $S_e$ so that the sum of chosen vectors along any root-to-leaf paths lies within certain convex bodies $K$.

\begin{theorem}\label{theorem:tree-reduction}
 Let $\calT = (V,E)$, $|E| \geq 2$, be a rooted tree such that: (1) every $e \in E$ has an associated set of vectors $S_e \subseteq \Ball^d_2$ satisfying $\mathbf{0} \in \conv(S_e)$, and (2) there exists an $\ell \in \mathbb{N}$, $\ell \geq 2$, such that, for all $e \in E$, $\mathbf{0}$ is a convex combination of at most $\ell$ vectors in $S_e$. Let $K \subseteq \mathbb{R}^d$ be a symmetric convex body with $\gamma_d(K) \geq 1 - \frac{1}{\ell\, |E|}$. Then, for every edge $e \in E$, there exists a vector $v_e \in S_e$, such that for all $u \in V$, $\sum_{e \in P_u} v_e \in 11 \, K$, where $P_u$ is the set of edges of the path from the root to the node $u$.
\end{theorem}
\begin{proof}
    Our goal is to select a vector $v_e \in S_e$ from each set $S_e$ to satisfy the desired property. At a high level, we start with a \emph{fractional} selection of vectors from each set $S_e$ such that, for every node $u \in V$, the fractional vector sum of the edges in the path from the root to $u$ is $\mathbf{0}$ (so, the desired property of being contained in $11 K$ is clearly satisfied). We iteratively round this fractional selection to get a single vector from each set $S_e$, in a way that every rounding step does not increase the vector sums we are interested in by too much.

    For all $e \in E$, by definition, $\mathbf{0} \in \conv(S_e)$ and $\mathbf{0}$ is a convex combination of at most $\ell$ vectors in $S_e$. 
    Therefore there exists a fractional selection of vectors $X^e = \{(v^e_1,x^e_1), (v^e_2,x^e_2), \ldots, (v^e_{\ell},x^e_{\ell})\}$, with $(v^e_i,x^e_i) \in S_e \times [0,1]$, such that $\sum_{i=1}^{\ell} x^e_{i} \cdot v^e_i = \mathbf{0}$, and $X^e$ is \emph{feasible}, i.e., $\sum_{i=1}^{\ell} x^e_{i} = 1$ and $x^e_i \in [0,1]$, for all $i \in [\ell]$. In the subsequent proof, we show the following two claims:
    \begin{enumerate}
        \item [(a)] 
            We can round each $x^e_i$ to $\hat{x}^e_i$ such that (i) for every $e \in E$, $\sum_{i=1}^{\ell} \hat{x}^e_i = 1$, and $\hat{x}^e_i \in [0,1]$ for all $i\in [\ell]$, (ii) the fractional part of each $\hat{x}^e_i$ is at most $\log(\frac{2 \, \ell \, h}{\varepsilon})$ bits long, where $h$ is the height of the tree $\calT$ and $\varepsilon = \frac{1}{2}-\frac{1}{\ell |E|}  > 0$, and (iii) for all $u \in V$, we have $\sum_{e \in P_u} \sum_{i=1}^{\ell} \hat{x}^e_i \cdot v^e_i \in K$.
        
        \item [(b)] 
            Let $Y^e = \{(v^e_1,y^e_1), \ldots, (v^e_{\ell},y^e_{\ell})\}$ be a feasible ($\sum_{i=1}^{\ell} y^e_i = 1$ and $y^e_i \in [0,1]$ for all $i \in [\ell]$), fractional selection of vectors. If, for every $e \in E$ and $i \in [\ell]$, the fractional part of $y^e_i$ is $k \geq 1$ bits long, then we can round the $k^{th}$-bit of all $\{y^e_i\}_{i,e}$ to get $\{\hat{y}^e_i\}_{i,e}$, whose fractional parts are at most $k-1$ bits long, and, for every $u \in V$, $\left( \sum_{e \in P_u} \sum_{i=1}^{\ell} \hat{y}^e_i v^e_i - \sum_{e \in P_u} \sum_{i=1}^{\ell} y^e_i v^e_i  \right) \in 2^{-k} \cdot 10 K$.
    \end{enumerate}

Starting from the fractional selection $X^e$, applying the rounding process of (a) gives us a feasible fractional selection $\hat{X}^e$ such that $\sum_{e \in P_u} \sum_{i=1}^{\ell} \hat{x}^e_i \cdot v^e_i \in K$. Then, repeatedly applying the rounding process of (b), starting from the fractional selection $\hat{X}^e$, results in an integral selection, after at most $\log(\frac{2 \, \ell \, h}{\varepsilon})$ steps. Let $A^e = \{ a^e_i \}_{i \in [\ell]}$, where $a^e_i \in \{ 0 , 1 \}$ and $\sum_{i=1}^\ell a^e_i = 1$, be the integral selection obtained at the end of the process. We have that $\left( \sum_{e \in P_u} \sum_{i=1}^{\ell} a^e_i v^e_i - \sum_{e \in P_u} \sum_{i=1}^{\ell} \hat{x}^e_i v^e_i  \right) \in \left( 2^{-\log(\frac{2 \, \ell \, h}{\varepsilon})} + 2^{-\log(\frac{2 \, \ell \, h}{\varepsilon})+1} + \dots + 2^{-1} \right) \cdot 10 K \subseteq 10 K$. And, since $\sum_{e \in P_u} \sum_{i=1}^{\ell} \hat{x}^e_i \cdot v^e_i \in K$, we have that $\sum_{e \in P_u} \sum_{i=1}^{\ell} a^e_i v^e_i \in 11 K$.

\paragraph{Proving $(a)$.} 
Let $\varepsilon = \frac{1}{2}-\frac{1}{\ell |E|}  > 0$ be a constant, and let $b = \log(\frac{2 \, \ell \, h}{\varepsilon})$. Construct $z^e_i$, for every $e \in E$ and $i \in [\ell]$, by taking $x^e_i$ and setting to zero all the bits in the fractional part of $x^e_i$ after the $b^{th}$ bit. 
Then, $\hat{x}^e_1 = z^e_1 + \left( 1 - \sum_{i=1}^{\ell} z^e_i \right)$, and $\hat{x}^e_i = z^e_i$ for all $i = 2, \dots, \ell$.
Clearly, for all $i \geq 2$, $\hat{x}^e_i \in [0,1]$, and the fractional part of $\hat{x}^e_i$ is at most $b$ bits long. 
By definition, $\sum_{i=1}^\ell \hat{x}^e_i = \hat{x}^e_1 + \sum_{i=2}^{\ell} \hat{x}^e_i = z^e_1 + \left( 1 - \sum_{i=1}^{\ell} z^e_i \right) + \sum_{i=2}^{\ell} z^e_i = 1$.
Furthermore, since $\sum_{i=2}^{\ell} \hat{x}^e_i \leq 1$ and $\sum_{i=1}^\ell \hat{x}^e_i = 1$, we have that $\hat{x}^e_1 \in [0,1]$. By the construction of the $z^e_i$s, the fractional part of $1 - \sum_{i=1}^{\ell} z^e_i = \sum_{i=1}^{\ell} x^e_i - \sum_{i=1}^{\ell} z^e_i$ is at most $b$ bits long, and therefore, the fractional part of $\hat{x}^e_1$ is at most $b$ bits long. Finally, for all $u \in V$,
\begin{align*}
\norm{ \sum_{e \in P_u} \sum_{i=1}^\ell \hat{x}^e_i \cdot v^e_i}_2 &= \norm{ \sum_{e \in P_u} \sum_{i=1}^\ell \hat{x}^e_i \cdot v^e_i - \sum_{e \in P_u} \sum_{i=1}^\ell x^e_i \cdot v^e_i}_2 \\
&= \norm{ \sum_{e \in P_u} \sum_{i=1}^\ell (\hat{x}^e_i - x^e_i) \cdot v^e_i }_2 \\
&\leq h \cdot \left( \ell \cdot 2^{-b} + (\ell - 1) \cdot 2^{-b} \right) \\
&\leq \varepsilon,
\end{align*}
where in the first equality we used the fact that $\sum_{e \in P_u} \sum_{i=1}^\ell x^e_i \cdot v^e_i = \mathbf{0}$ and in the first inequality we used that $\hat{x}^e_1 - x^e_1 \leq \ell \cdot 2^{-b}$. Therefore, $\sum_{e \in P_u} \sum_{i=1}^\ell \hat{x}^e_i \cdot v^e_i \in \varepsilon \, \Ball^d_2$. Since $\gamma_d(K) \geq 1 - \frac{1}{\ell|E|} \geq \frac{1}{2} + \varepsilon$, we have $\varepsilon \, \Ball^d_2 \subseteq K$ (\Cref{proposition:large-body-ball}). So, overall, $\sum_{e \in P_u} \sum_{i=1}^\ell \hat{x}^e_i \cdot v^e_i \in K$.
    
\paragraph{Proving $(b)$.} Assume that for all $e \in E$ and $i\in [\ell]$, the fractional part of $y^e_i$ is at most $k$ bits long. To round the $k^{th}$-bits of $\{y^e_i\}_{i,e}$, we first construct a new tree $\calT' = (V',E')$ that has vectors (instead of sets) associated to each edge, and then invoke~\Cref{theorem:tree-kulkarni} with $\mathcal{T}'$ and $K$. Intuitively, the signs from the guarantee of~\Cref{theorem:tree-kulkarni} will tell us how to round (up or down) the $k^{th}$-bit of $\{y^e_i\}_{i,e}$, so that the resulting $\{\hat{y}^e_i\}_{i,e}$ (after rounding) have at most $k-1$ bits in the fractional part, and this rounding process doesn't incur too much cost. 

For each $e \in E$, let $I^e = \{i_1, i_2, \ldots, i_{2q}\} \subseteq [\ell]$ be the set of indices such that, for every $j \in I^e$, the $k^{th}$ bit of the fractional part of $y^e_{j}$ is $1$; the set $I^e$ may be empty. Since $\sum_{i=1}^{\ell} y^e_i = 1$, $|I^e|$ must be even. For every $e \in E$, we pair up consecutive indices in $I^e$, and corresponding to each pair we define a vector. Formally, if $I^e = \emptyset$, set $\widetilde{S}^e = \{\mathbf{0}\}$; otherwise, if $I^e \neq \emptyset$, define the set of vectors 
\[
\widetilde{S}^e = \left\{ \frac{1}{2}\left(v^e_{i_{2p}} - v^e_{i_{2p+1}}\right) : \text{ for every } i_{2p}, i_{2p+1} \in I^e \right\}.
\]
For each $e \in E$, we have $\widetilde{S}^e \subseteq \Ball^d_2$ since $\norm{\frac{1}{2}\left(v^e_{i_{2p}} - v^e_{i_{2p+1}}\right)}_2 \leq \frac{1}{2}\left( \norm{v^e_{i_{2p}}}_2 + \norm{v^e_{i_{2p+1}}}_2\right) \leq 1$.  Also, by definition, $|\widetilde{S}^e| \leq \frac{|I^e|}{2} \leq \lfloor \frac{\ell}{2} \rfloor$.

To construct $\calT'$, we start from $\calT$ and replace every edge $e \in E$ by a path of $|\widetilde{S}^e|$ edges $e^{(1)}, e^{(2)}, \ldots, e^{(|\widetilde{S}^e|)}$. For every edge $e^{(i)} \in E'$, associate a unique vector $u^{(i)}_e \in \widetilde{S}^e$ (so, there's a bijection between $\widetilde{S}^e$ and $\{e^{(1)}, e^{(2)}, \ldots, e^{(|\widetilde{S}^e|)}\}$). Note that, since  $|\widetilde{S}^e| \leq \lfloor \frac{\ell}{2} \rfloor$, we have $|E'| \leq |E| \cdot \lfloor \frac{\ell}{2}\rfloor$. $\calT'$ satisfies the conditions of~\Cref{theorem:tree-kulkarni}: vectors associated with edges belong to the sets $\widetilde{S}^e$, where $\widetilde{S}^e \subseteq \Ball^d_2$. Applying~\Cref{theorem:tree-kulkarni} for $\calT'$ and $K$ (which also satisfies the conditions of~\Cref{theorem:tree-kulkarni}), there exists a sign $z_e^{(i)} \in \{-1,1\}$ for every $e \in E$ and $i \in [|\widetilde{S}^e|]$ (i.e., a sign for every edge in $E'$), such that, for every $u \in V$, $\sum_{e \in P_u} \sum_{i=1}^{|\widetilde{S}^e|} z_e^{(i)}u_e^{(i)} \in 5 K$ where $P_u$ is the set of edges on the path from the root node of $\calT$ to the node $u$ in $\calT$.

Consider an edge $e \in E$. We round every $y^e_i$ to $\hat{y}^e_i$ as follows: If $I^e = \emptyset$, then $\hat{y}^e_i = y^e_i$ for every $i \in [\ell]$. Otherwise, if $I^e \neq \emptyset$, then for every vector $u^{(i)}_e = \frac{1}{2}( v^e_{i_{2p}} - v^e_{i_{2p+1}} ) \in \widetilde{S}^e$, whose corresponding sign is $z^{(i)}_e$, we set $\hat{y}^e_{i_{2p}} =  y^e_{i_{2p}} + 2^{-k}\cdot z^{(i)}_e$ and $\hat{y}^e_{i_{2p+1}} = y^e_{i_{2p+1}} - 2^{-k}\cdot z^{(i)}_e$; all other $j \in [\ell] \setminus I^e$ are left unupdated, $\hat{y}^e_j = y^e_j$. This specifies a way to round each $y^e_{j}$ to $\hat{y}^e_{j}$.

It remains to show that our rounding procedure $(i)$ preserves feasibility, $(ii)$ sets the $k^{th}$-bit of the fractional part of $\hat{y}^e_i$ to $0$ for all $e$ and $i$, and $(iii)$ does not increase the vector sums of interest by too much: for all $u \in V$, 
$\left( \sum_{e \in P_u} \sum_{i=1}^{\ell} \hat{y}^e_i v^e_i - \sum_{e \in P_u} \sum_{i=1}^{\ell} y^e_i v^e_i  \right) \in 2^{-k} \cdot 10 K$.

Recall that $I^e$ is the set of all indices $i$ where the $k^{th}$ bit of the fractional part of $y^e_{i}$ is $1$. As per the aforementioned rounding process, for all $e \in E$ such that $I^e = \emptyset$, we have $\hat{y}^e_{i} = y^e_{i}$ for all $i \in [\ell]$, hence $(ii)$ holds.
Otherwise, if $I^e \neq \emptyset$, during the rounding we add or subtract $2^{-k}$ from every $y^e_{i_{2p}}$ and $y^e_{i_{2p+1}}$ respectively, where $i_{2p},i_{2p+1} \in I^e$. This addition and subtraction ensures that the $k^{th}$-bits of $y^e_{i_{2p}}$ and $y^e_{i_{2p+1}}$ are zero, additionally, for all $j \in [\ell] \setminus I^e$, $\hat{y}^e_j = y^e_j$, i.e., the $k^{th}$ bit remains zero; $(ii)$ follows. 
Since $\hat{y}^e_{i_{2p}} + \hat{y}^e_{i_{2p+1}} = y^e_{i_{2p}} + y^e_{i_{2p+1}}$, the equality $\sum_{i=1}^{\ell} \hat{y}^e_i = \sum_{i=1}^{\ell} y^e_i = 1$ is maintained. Additionally, $\hat{y}^e_j \geq 0$ for all $e, j$, the feasibility condition $(i)$, also holds.
Finally, for $(iii)$, recall that for all $u \in V$, $\sum_{e \in P_u} \sum_{i=1}^{|\widetilde{S}^e|} z_e^{(i)} u_e^{(i)} = \sum_{e \in P_u} \sum_{i=1}^{|\widetilde{S}^e|} z_e^{(i)} \frac{1}{2} \left(v^e_{i_{2p}} - v^e_{i_{2p+1}}\right) \in 5 K$, by~\Cref{theorem:tree-kulkarni}. Therefore, by our rounding process: $\sum_{e \in P_u} \sum_{i=1}^{\ell} (\hat{y}^e_i -  y^e_i ) \cdot v^e_i \leq \sum_{e \in P_u} \sum_{i=1}^{|\widetilde{S}^e|} 2^{-k} z^{(i)}_e \left(v^e_{i_{2p}} - v^e_{i_{2p+1}}\right) \in 2^{-k} \cdot 10 K$.
\end{proof}



Given~\Cref{theorem:tree-reduction}, our next task is to show that, given a rooted tree $\calT$ as above, there exists a distribution $\calD$ over vectors (one from each edge set) such that for $x \sim \calD$, $\sum_{e \in P_u} x_e$ is subgaussian, for every node $u \in V$.

\begin{theorem}\label{theorem:tree-subgaussianity}
    Let $\calT = (V,E)$ be a rooted tree, where every $e \in E$ has an associated set of vectors $S_e \subseteq \Ball^d_2$ satisfying $\mathbf{0} \in \conv(S_e)$. Then there exists a distribution $\calD$ supported on $\bigtimes_{e \in E} S_e$ such that for $x \sim \calD$, $\sum_{e \in P_u} x_e$ is $22.11$-subgaussian for every $u \in V$, where $P_u$ is the set of edges of the path from the root to the node $u$.
\end{theorem}


Finally, we prove that there exists an algorithm that, given sets of vectors one at a time, selects a vector from each set such that the vector sum is $O(1)$ subgaussian.

\begin{theorem}\label{theorem:subgauss-algo}
    For every $T, k \in \mathbb{N}$, there exists an online algorithm that, given sets $S_1, S_2, \ldots, S_T \subseteq \Ball^d_2$ satisfying $1 \leq |S_i| \leq k$ and $\mathbf{0} \in \conv(S_i)$, chosen by an oblivious adversary and arriving one at a time, selects a vector $s_i \in S_i$ from each arriving set $S_i$ such that, for every $t \in [T]$, the $\sum_{i=1}^t s_i$ is $23$-subgaussian. 
\end{theorem}

\begin{proof}
    Let $\beta = 22.11 = 23-\delta$ be the subgaussianity parameter in the guarantee of~\Cref{theorem:tree-subgaussianity}. Additionally, let $\calW$ be the smallest $\varepsilon$-net of the set $\calS = \cup_{i=1}^k \left( \bigtimes_{j=1}^i \Ball^d_2 \right)$ for $\varepsilon = \frac{\delta}{2T}$; here, $\calS$ represents the set of all subsets $A \subseteq \Ball^d_2$ satisfying $1 \leq |A| \leq k$. From~\Cref{proposition:size-of-net}, we know that $\calW$ has size at most $\sum_{i=1}^k \left(\frac{3\sqrt{i}}{\varepsilon}\right)^{di}   \leq \left(\frac{3\sqrt{k}}{\varepsilon}\right)^{d(k+1)}$. We consider a complete and full $|\calW|$-ary tree $\calT = (V,E)$ of height $T$, where every internal node $u \in V$ of $\calT$ has $|\calW|$ children, where each edge to a child-node is associated with an element of (or, set in) $\calW$. Let $A_e$ be the set that corresponds to edge $e \in E$ in our construction, where, by the definition of $\calW$, $A_e \subseteq \Ball^d_2$ and $1 \leq |A_e| \leq k$. ~\Cref{theorem:tree-subgaussianity} implies the existence of distribution $\calD$ over $\bigtimes_{e \in E} A_e$  such that for any node $u \in V$, $\sum_{e \in P_u} y_e$ is $\beta$-subgaussian for $y \sim \mathcal{D}$.
    At time $t=0$, we sample an $y \sim \calD$ and start at the root node of $\calT$. We will keep track of a location $p_t \in V$, which at the beginning of time $t$ will be a node at depth $t-1$. A time $t$, when the set $S_t$ arrives, we map it to a set $Y_t \in \argmin_{Z \in \calW \cap \Ball_2^{|S_t|}} \norm{Z - S_t}_2$, i.e., $Y_t$ is an element of $\calW$, with the same number of elements as $S_t$, closest to $S_t$. $Y_t$ corresponds to some edge $e_t$ incident to the current node $p_t$ (that is $A_{e_{t}} = Y_t$). Let $y_{t} \in Y_t$ be the vector corresponding to edge $e_t$ in the sample $y$ from $\calD$ (from time $0$). Given $y_{e_t}$, our algorithm selects vector $x_{t} \in \argmin_{v \in S_t}\norm{v - y_{e_t}}_2$. Overall, at time $t$ we: (i) map set $S_t$ to a set $Y_t$ (or, equivalently, edge $e_t$), (ii) use the sample $y$ (the same across all times) to identify a vector $y_{e_t} \in Y_t$, and finally (iii) map $y_{e_t}$ to a vector in $x_{e_t} \in S_t$; $x_{t}$ is our output in time $t$.
    Finally, we update $p_{t+1}$ to be the child of $p_t$ along edge $e_t$.

    Next, we prove that for all $t \in [T]$, $\sum_{i=1}^t x_t$ is $23$-subgaussian. Since $\calW$ is an $\varepsilon$-net of $\calS = \cup_{i=1}^k \left( \bigtimes_{j=1}^i \Ball^d_2 \right)$, $x_t \in S_t$ (and therefore, $x_t \in \calS$) and $y_{e_t} \in Y_t$ (and therefore $y_{e_t} \in \calW$), we have that that $\norm{x_t - y_{e_t}}_2 \leq \varepsilon$. Furthermore, $y \sim \calD$, and therefore $\sum_{e \in P_u} y_e$ is $\beta$-subgaussian for all $u \in V$. Noticing that $e_{1}, e_{2}, \dots, e_t$ form a path from the root of $\calT$ to some node $u$, we have that $\sum_{i=1}^t y_{e_i}$ is $\beta$-subgaussian, or, equivalently, $\norm{\sum_{i=1}^t y_{e_i}}_{\psi_2, \infty} \leq \beta = 23 - \delta$.
    
    Towards proving subgaussianity for $\sum_{i=1}^t x_i$ we have
    \begin{align*}
        &\norm{\sum_{i=1}^t x_i}_{\psi_2, \infty} \leq \norm{\sum_{i=1}^t y_{e_i}}_{\psi_2, \infty} + \norm{\sum_{i=1}^t x_i - y_{e_i}}_{\psi_2, \infty} \tag{triangle inequality}\\
        & \leq (23 - \delta) + \sum_{i=1}^t \norm{x_i - y_{e_i}}_{\psi_2, \infty} \tag{subgaussianity of $\sum_{i=1}^t y_{e_i}$ and triangle inequality}\\
        & \leq 23 - \delta + T \cdot \sup_{d \in \Sph^{d-1}} \norm{\langle x_i - y_{e_i}, d \rangle}_{\psi_2} \tag{definition of $\norm{.}_{\psi_2, \infty}$}\\
        &= 23 - \delta + T \cdot \sup_{d \in \Sph^{d-1}} \norm{ \, \norm{x_i - y_{e_i}}_2 \cdot \norm{d}_2 \cdot \cos(\theta) }_{\psi_2},
    \end{align*}
    where $\theta$ is the random angle between the vectors $x_i - y_{e_i}$ and $d$.  $\cos(\theta)$ is random variable supported on $[-1,1]$, $\norm{x_i - y_{e_i}}_2 \in [0,\varepsilon]$ and $\norm{d}_2 = 1$. Therefore, $\norm{x_i - y_{e_i}}_2 \cdot \norm{d}_2 \cdot \cos(\theta) \in [-\varepsilon,\varepsilon]$. From the definition of $\norm{.}_{\psi_2}$ norm we therefore have that $\norm{ \, \norm{x_i - y_{e_i}}_2 \cdot \norm{d}_2 \cdot \cos(\theta) }_{\psi_2} \leq \inf\{ t > 0: \E[exp(\varepsilon^2/t^2)] \leq 2 \} \leq 2 \varepsilon$. Our upper bound on $\norm{\sum_{i=1}^t x_i}_{\psi_2, \infty}$ becomes $23 - \delta + T \cdot 2 \varepsilon = 23$. 
\end{proof}

\subsection{Optimal online multicolor discrepancy}\label{subsec: optimal multicolor discrepancy}

Here, we prove our main result for this section, an optimal algorithm for online multicolor discrepancy. We start by giving an algorithm for \emph{weighted} online vector balancing.

\begin{restatable}{lemma}{InftyNormBound}\label{lemma:weighted-vector-balancing}
     For every $\alpha \in \left[\frac{1}{2}, \frac{2}{3}\right]$ and $T \in \mathbb{N}$, there exists an online algorithm that, given vectors $v_1, v_2, \ldots, v_T \in \Ball^d_2$ chosen by an oblivious adversary and arriving one at a time, assigns to each vector $v_i$ a weight $w_i \in \{1-\alpha, -\alpha\}$, such that, with probability at least $1 - \delta$, for any $\delta \in (0,1/2]$, for all $t \in [T]$, $\norm{\sum_{i=1}^t w_iv_i}_\infty \lesssim \sqrt{\log(T)} + \sqrt{\log(1/ \delta)}$.
\end{restatable}

Given an algorithm for weighted online vector balancing, we give an algorithm for online multicolor discrepancy: we construct a binary tree, where the leaves correspond to colors, and the internal nodes execute the weighted online vector balancing algorithm. We note that this trick has been used in the same context in previous work~\cite{alweiss2021discrepancy,bansal2021online}.

\begin{theorem}\label{thm: multi color main upper bound}
    For every $T \in \mathbb{N}$, there exists an online algorithm that, given vectors $v_1, v_2, \ldots, v_T \in \Ball^d_2$ chosen by an oblivious adversary and arriving one at a time, assigns each arriving vector $v_i$ to one of $n$ colors such that, with probability at least $1-\delta$, for any $\delta \in (0,1/2]$, for all $t \in [T]$,
    \[\max_{i,j \in [n]} \norm{\sum_{v \in \calC^t_i} v - \sum_{v \in \calC^t_j} v}_\infty \lesssim 6 \left( \sqrt{\log(T)} + \sqrt{\log(1/\delta)} \right) \]
 where $\calC^t_i$ is the set of all vectors that got assigned color $i \in [n]$ up to time $t \in [T]$.
\end{theorem}
\begin{proof}
    Let $\calA_\alpha$ be the algorithm of~\Cref{lemma:weighted-vector-balancing}, for an $\alpha \in [1/2, 2/3]$. We recursively construct a binary tree with $n$ leaves, corresponding to the $n$ colors. For a tree with $k > 1$ leaves we add a root node, as its left subtree recursively construct a tree with $\lceil k/2\rceil$ leaves, and as its right subtree recursively construct a tree with $\lfloor k/2\rfloor$ leaves; for $k=1$, we simply have a leaf. 
    
    Given vectors, one at a time, our algorithm for the online multicolor discrepancy problem decides which set/color a vector gets by repeatedly running $\calA_\alpha$ for the online vector balancing problem at each internal node of the tree. Specifically, at an internal node with $k$ descendent leaves, we will run a copy of $\calA_\alpha$ by setting $\alpha = \lceil k/2\rceil/k$, and by recursively passing the vectors that are assigned $1-\alpha$ (resp. $-\alpha$) to the left (resp. right) subtree (until they reach the leaves). Note that for any $k \geq 1$, we have $\alpha = \lceil k/2\rceil/k \in [1/2, 2/3]$. Vectors are assigned the color of the leaf they reach.

    Let $p_e$ be a weight for each edge $e$: the edge between the left (resp. right) child of an internal node with $k$ children has a weight $p_e = \alpha = \lceil k/2\rceil/k$ (resp. $p_e = 1 - \alpha$). The weights on the edges of an internal node is the ``opposite'' with respect to the weight of its children in the execution of $\calA_\alpha$. Intuitively, $p_e$ for an edge $(u,v)$ is the expected fraction of vectors that go to node $v$, out of the vectors that arrive at the parent node $u$.

    There are $n-1$ internal nodes in our tree. Let $\mathcal{E}$ be the event that all $n-1$ executions of $\calA_\alpha$ have maintained the discrepancy at most $\sqrt{\log(T)} + \sqrt{\log(1/ \delta)}$ between the corresponding two children nodes; $\mathcal{E}$ occurs with probability at least $1 - (n-1)\delta$. Let $S^{sum} = \sum_{i=1}^t v_i$ be the sum of all vectors until time $t$, and let $S_u$ be the sum of all vectors that have passed through node $u$ until time $t$ (so, $S_r = S^{sum}$ for the root node $r$). Also, let $\pi_u = \Pi_{e \in P_u} p_e$, for a node $u$; intuitively, $\pi_u$ is the (expected) fraction of vectors (out of $\{ v_1, \dots, v_t \}$) that arrive at node $u$.    

    We will prove, via induction on $\ell$, that conditioned on $\mathcal{E}$, for all nodes $u$ on level $\ell \leq 0$ we have
\[
\norm{ S_u - \pi_u S^{sum}}_{\infty} \lesssim 3 \left( \sqrt{\log(T)} + \sqrt{\log(1/ \delta)} \right).
\]

For $\ell = 0$ the statement trivially holds: for the root $r$ at level zero we have $\norm{ S_r - \pi_r S^{sum}}_{\infty} = 0 \lesssim 3 \left( \sqrt{\log(T)} + \sqrt{\log(1/ \delta)} \right)$. Suppose that the statement holds for level $\ell$, and let $u$ be a node in level $\ell + 1$, with parent node $p$ (on level $\ell$) and sibling node $v$ (on level $\ell + 1$). Assume that $u$ is the left child of $p$ (the other case is identical). We have that $S_p = S_u + S_v$, and $\pi_u = \pi_p \cdot \alpha$. Also, conditioned on $\mathcal{E}$ we have $\norm{ (1-\alpha)S_u - \alpha S_v}_{\infty} \lesssim \sqrt{\log(T)} + \sqrt{\log(1/ \delta)}$. So, overall:
\begin{align*}
    &\norm{ S_u - \pi_u S^{sum}}_{\infty} = \norm{ (1-\alpha)S_u + \alpha S_u - \alpha \pi_p S^{sum}}_{\infty} \\
    &\leq \norm{ (1-\alpha)S_u - \alpha S_v }_{\infty} + \norm{ \alpha S_v + \alpha S_u - \alpha \pi_p S^{sum}}_{\infty} \tag{triangle inequality}\\
    &= \norm{ (1-\alpha)S_u - \alpha S_v }_{\infty} + \alpha \norm{ S_p - \pi_p S^{sum}}_{\infty} \\
    &\lesssim \left( \sqrt{\log(T)} + \sqrt{\log(1/ \delta)} \right) + \frac{2}{3} \, 3 \left( \sqrt{\log(T)} + \sqrt{\log(1/ \delta)} \right) \\
    &= 3 \left( \sqrt{\log(T)} + \sqrt{\log(1/ \delta)} \right).
\end{align*}

We will also prove, via induction on $k$, that for a node $u$ that is the root of a subtree with $n-k+1$ leaves, $\pi_u = (n-k+1) \cdot \frac{1}{n}$. For the root $r$ (whose subtree has $n = n-1+1$ leaves) we have $\pi_r = 1 = n \cdot \frac{1}{n}$. Consider a node $u$ that is the left child of a node $p$, such that $p$ is the root of a subtree with $k$ leaves. Then $u$ is the root of a subtree with $\lceil k/2\rceil$ leaves, and $\pi_u = \pi_p \cdot \alpha = k \, \frac{1}{n} \cdot \lceil k/2\rceil/k = \lceil k/2\rceil \cdot \frac{1}{n}$; the case that $u$ is the right child of $p$
 is identical.

Finally, consider two arbitrary leaves $v_1$ and $v_2$. From the previous arguments we have that $\pi_{v_1} = \pi_{v_2} = \frac{1}{n}$, and $\norm{ S_{v_i} - \frac{1}{n} \, S^{sum}}_{\infty} \lesssim 3 \left( \sqrt{\log(T)} + \sqrt{\log(1/ \delta)} \right) $. Therefore, $\norm{ S_{v_1} - S_{v_2}}_{\infty} \lesssim 6 \left( \sqrt{\log(T)} + \sqrt{\log(1/ \delta)} \right)$.
\end{proof}

The lower bound for the online envy minimization problem in the next section implies that the bound of~\Cref{thm: multi color main upper bound} is optimal.

\subsection{Optimal online envy minimization}\label{subsec: envy for oblivious}

\Cref{thm: multi color main upper bound} immediately implies, for the online envy minimization problem, a $O_n(\sqrt{\log{T}})$ upper bound against an oblivious adversary.

\begin{corollary}\label{cor: main result for oblivious and fair division}
    For any $n \geq 2, T \geq 1$ and $\delta \in (0,1/2]$, there exists an online algorithm that, given a sequence of $T$ items with $v_{i,t} \in [0,1]$ for all $i \in [n]$ and $t \in [T]$ selected by an oblivious adversary and arriving one at a time, allocates each item to an agent such that the envy between any pair of agents $i,j \in [n]$ satisfies, $\envy^t_{i,j} \in O_n(\sqrt{\log{T}})$ with probability at least $1 - \frac{1}{T^c}$, for any constant $c$.
\end{corollary}

Here, we prove a lower bound of $\Omega_n((\log(T))^{r/2})$, for all $r < 1$, for the online envy minimization problem, against an oblivious adversary. Our proof crucially uses the construction in the lower bound of~ \cite{benade2024fair} for the online envy minimization problem, against an adaptive adversary; for completeness, we include a proof of this result in~\Cref{app:proof from OR paper}.

\begin{theorem}[Theorem 2 of \cite{benade2024fair} restated]\label{theorem:adaptive-lb}
    For any $n \geq 2$, $r < 1$ and $T \geq 1$, there exists a set $S_T$ of instances with $|S_T|  \leq 2^T$ such that for any online algorithm $\mathcal{A}$, there exists an instance $I \in [0,1]^{n\cdot T}$ in $S_T$ such that running algorithm $\mathcal{A}$ on the sequence of items $1,2 \ldots, T$ described by $I$ results in a maximum envy of at least $\envy^T \in \Omega_n(T^{r/2})$ at time $T$.
\end{theorem}

Our lower bound is stated as follows. For any $c > 0$, by setting $\delta = 1/T^c$ in the theorem statement below, we obtain the corollary that any online algorithm must result in $\envy^T \in \Omega_n((\log{T})^{r/2})$ with probability at least $1/T^c$.

\begin{theorem}\label{thm: envy lower bound for oblivious}
    Fix any $n \geq 2$, $T \geq 1$, $r \in (0,1)$, and $\delta \in (0,1)$. Let $\mathcal{A}$ be a (possibly randomized) online algorithm. There exists an oblivious adversary that can select a sequence of $T$ items such that the allocation $A^T$ constructed by $\mathcal{A}$ has $\envy^T \in \Omega_n((\log{\frac{1}{\delta}})^{r/2})$ with probability at least $\delta$.
\end{theorem}
\begin{proof}
    By Yao's minimax principle, we can, without loss of generality, focus on deterministic algorithms $\mathcal{A}$ and an adversary that selects distributions over instances. We will construct a distribution $\mathcal{D}$ over instances, such that any deterministic algorithm $\mathcal{A}$ has $\envy^T \in \Omega_n((\log{\frac{1}{\delta}})^{r/2})$ with probability at least $\delta$, where the randomness is over instances drawn from $\mathcal{D}$. $\mathcal{D}$ is defined as follows: for a fixed $T$, consider the set of instances $S_{\log{\frac{1}{\delta}}}$ described in~\Cref{theorem:adaptive-lb}, and select an instance uniformly at random from this set. This gives us a random sequence of $\log{\frac{1}{\delta}}$ items; to get to $T$ items, include $T - \log{\frac{1}{\delta}}$ items that have a zero value for all the agents. Note that, by definition, $S_{\log{\frac{1}{\delta}}}$ contains an instance $I^*$ for which algorithm $\mathcal{A}$ incurs an maximum envy of $\Omega_n((\log{\frac{1}{\delta}})^{r/2})$ at time $\log{\frac{1}{\delta}}$, and therefore at time $T$ as well, since all items after step $\log{\frac{1}{\delta}}$ have zero value. Finally, note that $\mathcal{D}$ samples $I^*$ with probability exactly $1/|S_{\log{\frac{1}{\delta}}}|$, which is at least $1/2^{\log{\frac{1}{\delta}}} = \delta$, by~\Cref{theorem:adaptive-lb}.
\end{proof}

\section{Performance Against an i.i.d.\@ Adversary}\label{sec: iid}

In this section, we study an i.i.d.\@ adversary. In this model, we show that online envy minimization is easier than online multicolor discrepancy.
We first prove a super-constant lower bound for the online vector balancing problem (\Cref{thm: lower bound for iid vector balancing}), which, naturally, implies a super-constant lower bound for the online multicolor discrepancy problem. In~\Cref{subsec: iid envy n agents} we give a simple algorithm for online envy minimization and $n$ agents. All missing proofs can be found in~\Cref{app:missing proofs from iid}.

\subsection{Lower bounds for online vector balancing}

In the following lower bound, we show that if for all $t \in [T]$, each coordinate of all the vectors $v_t$ are i.i.d.\@ drawn from the distribution $\mathcal{U}([-1,1])$, then the discrepancy at time $T$ of any online algorithm must be $\Omega\left(\sqrt{\frac{\log{T}}{\log{\log{T}}}}\right)$. Note that a drawn vector might not be a member of $\mathcal{B}_2^d$. However, the same lower bound will hold up to a factor of $\sqrt{d}$ if each coordinate is drawn from $\mathcal{U}([-1/\sqrt{d},1/\sqrt{d}])$ which ensures $v_t \in \mathcal{B}_2^d$; we use $\mathcal{U}([-1,1])$ for the ease of exposition.

\begin{theorem}\label{thm: lower bound for iid vector balancing}
    Even for $n=2$ colors, for any $T \in \mathbb{N}$, any online algorithm $\mathcal{A}$, and any $d > 2$, when $\mathcal{A}$ is presented with a sequence of vectors $v_1, \dots, v_T \in \mathbb{R}^d$, where $v_{t,i} \sim \mathcal{U}([-1,1])$ in an i.i.d.\@ fashion, the discrepancy of $\mathcal{A}$ is $\Omega\left(\sqrt{\frac{\log{T}}{\log{\log{T}}}}\right)$, with probability at least $1 - 1/T^{\Theta(1)}$. 
\end{theorem}
\begin{proof}
    Let $\dist = \mathcal{U}([-1,1])$. We use the notation $v \sim \dist^d$ to denote a random vector $v \in \mathbb{R}^d$ each of whose coordinates are drawn independently from $\dist$. 
    The key observation is that, with sufficiently high probability, there is a long enough sequence of input vectors that are orthogonal to the current discrepancy vector; this leads to a large discrepancy at the end of this sequence. The following claim will be used to formalize this idea. 

    \begin{claim}\label{claim:hypercube_prob}
        There exists a constant $c>0$ such that, for all $\delta > 0$, and $u \in \mathbb{R}^d$, we have $\Pr_{v \sim \dist^d}[|\langle v, u\rangle| \leq \delta \norm{u}_2 \text{ and } \norm{v}_2 \in [1/2,1]] \geq c \delta$, where the constant $c$ depends on $d$. 
    \end{claim}
    \begin{proof} We can rewrite the probability of interest as
    \begin{align*}
        & \phantom{{}={}} \Pr_{v \sim \dist^d}[|\langle v, u\rangle| \leq \delta \norm{u}_2 \text{ and } \norm{v}_2 \in [1/2,1]]\\
        & = \Pr_v[\norm{v}_2 \in [1/2,1]] \cdot \Pr_v[|\langle v, u\rangle| \leq \delta \norm{u}_2 \mid \norm{v}_2 \in [1/2,1]] \numberthis \label{equation:obv-iid-lower-bound}
    \end{align*}

    We will show that $\Pr_v[\norm{v}_2 \in [1/2,1]] \geq c_1$ and $\Pr_v[|\langle v, u\rangle| \leq \delta \norm{u}_2 \mid \norm{v}_2 \in [1/2,1]] \geq c_2 \delta$ where $c_1$ and $c_2$ are constants that depends on $d$. These two inequalities, along with \Cref{equation:obv-iid-lower-bound}, imply that $\Pr_{v \sim \dist^d}[|\langle v, u\rangle| \leq \delta \norm{u}_2 \text{ and } \norm{v}_2 \in [1/2,1]] \geq c_1c_2 \delta = c \delta$, where $c = c_1 c_2$.

    To prove that $\Pr_v[\norm{v}_2 \in [1/2,1]] \geq c_1$ we use the fact that the volume of the unit Euclidean ball is given by $\textrm{vol}(\mathcal{B}^d_2) = \frac{\pi^{d/2}}{\Gamma(d/2+1)}$ where $\Gamma$ represents the gamma function~\cite{smith1989small}:
    $\Pr[\norm{v}_2 \in [1/2,1]] = \frac{\textrm{vol}(\mathcal{B}^d_2) - \textrm{vol}(\mathcal{B}^d_2)/2^d}{2^d} \geq \frac{\textrm{vol}(\mathcal{B}^d_2)}{2^{d+1}} = c_1,$
    where $c_1$ only depends on $d$. 
    
    It remains to prove that $\Pr_v[|\langle v, u\rangle| \leq \delta \norm{u}_2 \mid \norm{v}_2 \in [1/2,1]] \geq c_2 \delta$ for a constant $c_2$ that depends only on $d$. Conditioning on the event $\norm{u}_2 \in [1/2,1]$, the distribution of the random vector $v \sim \mathcal{D}^d$ is centrally symmetric, i.e., the probability density of $v$ only depends on $\norm{v}_2$ and not the direction of $v$. Define $\theta$ to be the random angle between $u$ and $v$. All possible angles $\theta \in [0, 2\pi]$ that $u$ can make with $v \sim \dist^d$ are equally likely. Using this fact, we get
        \begin{align*}
 &\phantom{{}={}}\Pr_{v \sim \dist^d}[|\langle v, u\rangle| \leq \delta \norm{u}_2 \mid \norm{v}_2 \in [1/2,1]] \\
 & =  \Pr_{v \sim \dist^d}[ | \sqrt{d} \cos{\theta} | \leq \delta  \mid \norm{v}_2 \in [1/2,1]]\\
  & =  \Pr_{v \sim \dist^d}[ | \cos{\theta} | \leq \frac{\delta}{\sqrt{d}}  \mid \norm{v}_2 \in [1/2,1]]\\
  & =  \frac{(\pi/2 - \arccos(\delta/\sqrt{d}))}{\pi/2} \tag{$\theta \in [0, 2\pi]$ is uniformly distributed} \\
  & \geq 1 - \frac{\arccos(\delta/\sqrt{d})}{\pi/2} = \frac{2}{\pi\sqrt{d}} \cdot \delta,
        \end{align*}
        the penultimate inequality here follows from the Taylor expansion of $\arccos$, which implies that $\arccos(x) \leq \pi/2 - x$ for $x \geq 0$. Setting $c_2 = \frac{2}{\pi\sqrt{d}}$ completes the proof of the claim.
    \end{proof}
    Denote by $d_t \coloneqq \sum_{i=1}^t \chi_i v_i$, where $\chi_i \in \{ -1, 1 \}$ is the sign the algorithm picks, the discrepancy at time $t$. We know that, $\norm{d_t}_2^2 \geq \norm{d_{t-1}}_2^2 + \norm{v_t}_2^2 - 2|\langle d_{t-1}, v_t \rangle|$. For the case when $\norm{d_{t-1}}_2 \leq \frac{1}{8 \delta}$, from \Cref{claim:hypercube_prob}, we have that, with probability at least $c \delta$, $|\langle d_{t-1}, v_t\rangle| \leq 1/8$ and $\norm{v_t}_2 \in [1/2,1]$. Both these events imply that $\norm{d_t}_2^2 \geq \norm{d_{t-1}}_2^2 + 1/2 - 2\cdot 1/8 = \norm{d_{t-1}}_2^2 + 1/4.$
    
    We now divide the time horizon from $1,\ldots, T$ into $T/\tau$ contiguous chunks having $\tau$ timesteps each. Consider a contiguous chunk spanning timesteps $t_s, \ldots, t_e$ where $t_e - t_s = \tau$. Note that with probability at least $(c\delta)^\tau$ all the incoming vectors in this chunk will satisfy the condition in \Cref{claim:hypercube_prob}, thereby implying that $\norm{d_{t_e}}_2^2 - \norm{d_{t_s}}_2^2 \geq \tau/4$, which in turn will imply that $\norm{d_{t_e}}_2 \geq \sqrt{\tau/4}$.
    
    We now set $\delta = 1/(c \log{T})$ and $\tau = \log{T}/(2\log{\log{T}})$. Either at some point we have $\norm{d_{t-1}}_2 > \frac{1}{8 \delta} = (c\log{T})/8$, in which case the lower bound holds. Otherwise $\norm{d_{t-1}}_2 < \frac{1}{8 \delta}$ for all the timesteps, and with probability at least $1- \left(1 - (c\delta)^\tau\right)^{T/\tau} = 1- \left(1 - (1/\log{T})^{\log{T}/(2\log{\log{T}})}\right)^{(2T\log{\log{T}})/\log{T}} = 1- 1/T^{\Theta(1)}$ at least one of the chunks will have all its vectors almost orthogonal to the current discrepancy vector (i.e., all vectors will satisfy the condition in \Cref{claim:hypercube_prob}), leading to a discrepancy of at least $\sqrt{\tau/4} = O\left(\sqrt{\frac{\log{T}}{\log{\log{T}}}}\right)$. This concludes the proof of~\Cref{thm: lower bound for iid vector balancing}.
\end{proof}

\subsection{Online envy minimization}\label{subsec: iid envy n agents}

In this section, we give an algorithm,    ~\Cref{algo:online envy for n agents}, for online envy minimization, against an i.i.d.\@ adversary. \Cref{algo:online envy for n agents} works in two phases. In phase 1, which lasts $T^{(1)}$ steps, it makes allocations using the welfare maximization algorithm (``item $j$ is allocated to the agent with the largest value''). In Phase 2, at every step $t$ the algorithm singles out the set of agents who have not received a large number of items (within phase 2, up until $t$); among this set, it allocates item $t$ to the agent who is envied the least by agents in this set.

\begin{theorem}\label{thm:n agent upper bound iid fair division}
For all positive integers $c$, \Cref{algo:online envy for n agents} has envy at most $c + 1$ with probability at least $1 - O(T^{-c/2})$.
\end{theorem}

Note that the $O(\cdot)$ hides constants that depend on the number of agents, but is independent of the value distribution.

\begin{proof}
    Fix a positive integer $c$, an arbitrary distribution $\calD$ supported on $[0, 1]$, and a time horizon $T$. Throughout, we assume that $T$ is sufficiently large, i.e., larger than some number $T_0$ that depends only on $n$ and $c$, not on the distribution $\calD$. Let $F$ denote the CDF of $\calD$.

    A key observation is that we can analyze the algorithm using an equivalent, but more structured method of sampling item values. Normally, at each time step $t$, the item values revealed to the algorithm are sampled i.i.d.\@ from $\mathcal{D}$, independent of every decision made so far. Instead, we define an equivalent experiment as follows. Let $G^{welf} = \{g^{welf}_1, \ldots, g^{welf}_{T^{(1)}}\}$ be a set of $T^{(1)}$ goods and, for each agent $i \in [n]$,  let $G^{i} = \{g^i_1, \ldots, g^i_{T^{(2)}}\}$ be a set of $T^{(2)}$ goods.
    \begin{enumerate}[leftmargin=*]
        \item Before the algorithm begins, nature samples values $(V^g_1, \ldots, V^g_n)$ for each $g \in G^{welf} \cup \bigcup_i G^i$, where each $V^g_i \stackrel{i.i.d.}{\sim} \calD$. 
        \item During \emph{Phase 1} of the algorithm (welfare maximization), when the $t^{th}$ item arrives, it is revealed to be  item $g^{welf}_t$, with pre-sampled values $(V^{g^{welf}_t}_1, \ldots, V^{g^{welf}_t}_n)$.
        \item During \emph{Phase 2} (lines 3-7 in~\Cref{algo:online envy for n agents}), suppose item $t$ will be assigned to agent $i$ who, at this point, has received $k$ items during phase $2$ ($|A^t_i \setminus G^{welf}| = k$). Then, item $t$ is revealed to be $g^i_{k + 1}$ with pre-sampled values $(V^{g^i_{k + 1}}_1, \ldots, V^{g^i_{k + 1}}_n)$.
    \end{enumerate}

    \begin{algorithm}[t]
\caption{Two-Phase Envy Minimization Algorithm}\label{algo:online envy for n agents}
\SetAlgoLined
\DontPrintSemicolon

Set $T^{(1)} \gets T - \frac{n(n - 1)}{2}\lceil \log T \sqrt{T} \rceil$, and $T^{(2)} \gets \frac{n(n - 1)}{2} \lceil \log T \sqrt{T} \rceil$\;
Run welfare maximization (i.e., allocate item $t$ to $\argmax_{i \in [n]} v_{i,t}$ breaking ties randomly)  for $T^{(1)}$ steps \;
\For{$t \gets T^{(1)} + 1$ \KwTo $T$}{
    Let $w^t_i$ be the number of items agent $i$ has received in steps $t' >T^{(1)}$.\;
    Let $S$ be the smallest (in terms of cardinality) subset of agents, such that $\forall i \in S, j \notin S$ $w^t_i \leq w^t_j - \lceil \log T \sqrt{T} \rceil$. \;
    Allocate item $t$ to an agent $i \in S$ who is envied the least, i.e., $\argmin_{i \in S} \max_{j \in S} \envy^t_{j,i}$. \;
}
\end{algorithm}

    Importantly, the allocation decision for item $t$ does not depend on agents' values for this item. This ensures that the value vector $(V^{g^i_{k + 1}}_1, \ldots, V^{g^i_{k + 1}}_n)$ is independent of all decisions made by the algorithm. Consequently, this modified experiment is statistically identical to the original setup in terms of the envy of the final allocation.

    A second useful modification is to work with item \emph{quantiles} instead of item values. More formally, instead of directly sampling $V^g_i$, we will first sample a quantile $Q^g_i \sim \mathcal{U}[0, 1]$ and then set $V^g_i = F^{-1}(Q^g_i)$ where $F^{-1}$ is the generalized inverse of $F$. Throughout the remainder of this proof, we condition on the probability $1$ event that all $Q^g_i$s are distinct.
    Note that for $g \in G^{welf}$, allocating item $g$ to an agent with the highest quantile, $i \in \argmax_{j} Q^g_j$, is equivalent to welfare maximization with random tie-breaking.\footnote{We make this point, since unequal quantiles does not imply unequal values.} Thus, we will assume these are coupled. Since all quantiles are distinct by assumption, ties never occur, and this allocation is always well-defined. 

    \textbf{No heavy envy-cycles.} Our first high-level step will be to show that, with high probability, no envy cycles with large weight exist during the execution of the algorithm.

    \begin{lemma}\label{lem:no-cycle}
        With probability $1 - O ( T^{-c/2}  )$, at every time $t \ge T^{(1)}$, there does not exist a cycle of agents $i_1, \ldots, i_k, i_{k+1}=i_1$ such that $\envy_{i_j, i_{j + 1}} > c$ for all $j = 1, \ldots, k$. 
    \end{lemma}

    The proof of~\Cref{lem:no-cycle} crucially relies on the following concentration inequality (which, to the best of our knowledge, is not known), that might be of independent interest.

    \begin{lemma}\label{lem:concentration}
    Fix positive integers $L, K$, and $c$, with $L < \frac{K}{4e}$. Let $Y_1, \ldots, Y_K$ be i.i.d.\@ draws from a distribution supported on $[0, 1]$. Then, $\Pr\left[\sum_{i \le K - L} Y_i - \sum_{i > K - L} Y_i < -c \right] \leq 4 \cdot \left(\frac{2 e L}{K} \right)^{c + 1}$.
    \end{lemma}

    The proof of~\Cref{lem:no-cycle}  also relies on two (relatively more straightforward) facts,~\Cref{lem:phase-2-items,lem:halls}. 
    The first lemma shows that the items allocated in phase 2 are relatively balanced among the agents, up to additive $\ceil{\log T \sqrt{T}}$ factors.

    \begin{lemma}\label{lem:phase-2-items}
        Fix a time $t$. Let $w^t_i|$ be the number of items agent $i$ has received in phase 2, i.e., $w^t_i = |A^T_i \setminus G^{welf}|$. Let $(w^t_{i_1}, \ldots, w^t_{i_n})$ be these numbers sorted from smallest to largest; so, $w^t_{i_j} \le w^t_{i_{j + 1}}$. Then, for all $j \le n - 1$, $w^t_{i_{j + 1}} \le w^t_{i_j} + \ceil{\log T \sqrt{T}}$. Furthermore, $w^t_{i_n} \le (n - 1) \ceil{\log T \sqrt{T}}$, i.e., no agent ever receives more than $(n - 1) \ceil{\log T \sqrt{T}}$ phase 2 items. 
    \end{lemma}
    The second lemma is a sufficient condition for bounding the envy between two sets of values and is reminiscent of approximate \emph{stochastic-dominance envy-freeness} (SD-EF). The proof is based on a generalization of Hall's theorem.
    \begin{lemma}\label{lem:halls}
        Given two sequences of values $a_1, \ldots, a_k$ and $b_1, \ldots, b_{\ell}$ where each $a_i, b_i \in [0, 1]$. Suppose that, for each $a_i$, $|\set{i' | a_{i'} \ge a_i}| \le |\set{i' | b_{i'} \ge a_i}| + c$. Then, $\sum_i a_i \le \sum_i b_i + c$. 
    \end{lemma}

    \textbf{Long phase 2 eliminates envy.} Our second high-level step will be to show that if phase 2 is sufficiently long, then with high probability, envy can be eliminated.


    \begin{lemma}\label{lem:high-value-n} With probability $1 - O(T^{-c/2})$ it is the case that for all agents $i, j \in \agents$ and all time steps $t \ge T^{(1)}$, if $|A^t_i \setminus G^{welf}| \ge |A^t_j \setminus G^{welf}| + \ceil{\log T \sqrt{T}}$, then $\envy^t_{i, j} \le c$. 
    \end{lemma}

    \textbf{Putting it all together.}
    With~\Cref{lem:no-cycle,lem:high-value-n} in hand, we are ready to prove the theorem.
    
    Let $H^t$ be a graph with nodes $[n]$ where there is an edge $(i, j)$ if $\envy^t_{i, j} > c$. Condition on the events in \Cref{lem:no-cycle} and \Cref{lem:high-value-n}. 
    These happen with probability $1 - O(T^{-c/2})$. Then,~\Cref{lem:no-cycle} ensures that $H^t$ is acyclic for all $t$, while~\Cref{lem:high-value-n} ensures that if $|A^t_i \setminus G^{welf}| \le  |A^t_j \setminus G^{welf}| + \ceil{\log T \sqrt{T}}$, then $(j, i) \notin H^t$.

    First, we prove that if an agent $i$ received an item at some point during phase 2, then $\envy^T_{j, i} \le c + 1$, for all $j \in \agents$. To this end, suppose that $i$ does indeed receive an item at some point during phase 2, and let $t$ be the \emph{last} time step for which $i$ received an item. 

    We first claim that $\envy^{t - 1}_{j, i} \le c$ for all $j \ne i$, i.e., $i$ is a source node in $H^{t - 1}$. Let $S$ be the set defined in~\Cref{algo:online envy for n agents} for time step $t$, i.e., for all $j \in S$ and $j' \notin S$, $j$ has received at least $\ceil{\log T \sqrt{T}}$ fewer items in phase 2, or  $|A^{t - 1}_j \setminus G^{welf}| \le |A^{t - 1}_{j'} \setminus G^{welf}| - \ceil{\log T \sqrt{T}}$. Note that $i \in S$ as they received item $t$, and for all $j' \notin S$, $(j', i) \notin H^{t - 1}$, by~\Cref{lem:high-value-n}. Then, consider $H^{t - 1}[S]$, the subgraph of $H^{t - 1}$ that only containing the nodes $S$. Note that $H^{t - 1}[S]$ is acyclic, since $H^{t - 1}$ is acyclic. Furthermore, by definition, source nodes in $H^{t - 1}[S]$ are envied by $\le c$ by all agents in $S$, while non-source nodes are envied by $>c$ by at least one agent in $S$. Hence, $i$ must be a source node in $H^{t - 1}[S]$. Together with the fact that there are no $(j, i)$ edges for $j \notin S$, we have that $i$ is a source node in $H^{t - 1}$.

    Now, since $\envy^{t - 1}_{j, i} \le c$, giving an item to $i$ can increase envy by at most $1$. Hence, $\envy^{t}_{j, i} \le c + 1$. Furthermore, as $i$ never received any more items after this time ($t$ was defined as the last item $i$ received), envy toward $i$ cannot increase. Hence, $\envy^T_{j, i} \le c + 1$, as needed. 

    Finally, let $w^t_i = |A^T_i \setminus G^{welf}|$ be the number of items allocated to agent $i$ in phase $2$. Note that if $w^t_i > 0$, by our previous argument, $\envy_{j, i} \le c + 1$ for all $j \ne i$. So, if $w^t_i > 0$ for all $i \in \agents$, we are done. Suppose this is not the case. So, there exists an agent $i$ such that  $w^t_i = 0$. Let $i_1, \ldots, i_n$ be an ordering of the agents sorted by $w^t_i$, i.e., $w^t_{i_1} \le \cdots \le w^t_{i_n}$. The assumption that $w^t_i = 0$ implies that $w^t_{i_1} = 0$. We claim that $w^t_{i_2} \ge \ceil{\log T \sqrt{T}}$. Indeed, by induction,~\Cref{lem:phase-2-items} ensures that for all $j \ge 2$, $w^t_{i_j} \le w^t_{i_2} + (j - 2) \ceil{\log T \sqrt{T}}$. Hence, 
    $\sum_{j=1}^n w^t_{i_j} \le (n - 1) \cdot w^t_{i_2} + \frac{(n - 2)(n - 1)}{2} \, \ceil{\log T \sqrt{T}}$. 
    However, since this is time $T$,  $\sum_j w^t_{i_j} = T^{(2)} = \frac{n(n - 1)}{2} \cdot \ceil{\log T \sqrt{T}}$. Together, these imply that $w^t_{i_2} \geq \ceil{\log T \sqrt{T}}$. Therefore, all agents other than $i_1$ received at least one item during phase 2, and hence are not envied by more than $c + 1$. On the other hand, all agents $j \ne i_1$ received at least $\ceil{\log T \sqrt{T}}$ items more than $i_1$ in phase $2$, and therefore, $\envy_{j, {i_1}} \le c \le c + 1$ as needed.
\end{proof}

\bibliographystyle{plainnat}
\bibliography{abb,references}

\begin{thebibliography}{55}
\providecommand{\natexlab}[1]{#1}
\providecommand{\url}[1]{\texttt{#1}}
\expandafter\ifx\csname urlstyle\endcsname\relax
  \providecommand{\doi}[1]{doi: #1}\else
  \providecommand{\doi}{doi: \begingroup \urlstyle{rm}\Url}\fi

\bibitem[Aleksandrov et~al.(2015)Aleksandrov, Aziz, Gaspers, and
  Walsh]{aleksandrov2015online}
Martin Aleksandrov, Haris Aziz, Serge Gaspers, and Toby Walsh.
\newblock Online fair division: analysing a food bank problem.
\newblock In \emph{Proceedings of the 24th International Joint Conference on
  Artificial Intelligence (IJCAI)}, pages 2540--2546, 2015.

\bibitem[Alweiss et~al.(2021)Alweiss, Liu, and Sawhney]{alweiss2021discrepancy}
Ryan Alweiss, Yang~P Liu, and Mehtaab Sawhney.
\newblock Discrepancy minimization via a self-balancing walk.
\newblock In \emph{Proceedings of the 53rd Annual ACM Symposium on Theory of
  Computing (STOC)}, pages 14--20, 2021.

\bibitem[Amanatidis et~al.(2017)Amanatidis, Markakis, Nikzad, and
  Saberi]{amanatidis2017approximation}
Georgios Amanatidis, Evangelos Markakis, Afshin Nikzad, and Amin Saberi.
\newblock Approximation algorithms for computing maximin share allocations.
\newblock \emph{ACM Transactions on Algorithms}, 13\penalty0 (4):\penalty0
  1--28, 2017.

\bibitem[Aru et~al.(2018)Aru, Narayanan, Scott, and
  Venkatesan]{aru2016balancing}
Juhan Aru, Bhargav Narayanan, Alexander Scott, and Ramarathnam Venkatesan.
\newblock Balancing sums of random vectors, 2018.

\bibitem[Bai and G{\"{o}}lz(2022)]{bai2021envy}
Yushi Bai and Paul G{\"{o}}lz.
\newblock Envy-free and pareto-optimal allocations for agents with asymmetric
  random valuations.
\newblock In \emph{Proceedings of the 31st International Joint Conference on
  Artificial Intelligence (IJCAI)}, pages 53--59, 2022.

\bibitem[Ball(1997)]{ball1997elementary}
Keith Ball.
\newblock An elementary introduction to modern convex geometry.
\newblock \emph{Flavors of geometry}, 31\penalty0 (1-58):\penalty0 26, 1997.

\bibitem[Banaszczyk(2012)]{banaszczyk2012series}
Wojciech Banaszczyk.
\newblock On series of signed vectors and their rearrangements.
\newblock \emph{Random Structures \& Algorithms}, 40\penalty0 (3):\penalty0
  301--316, 2012.

\bibitem[Banerjee et~al.(2022)Banerjee, Gkatzelis, Gorokh, and
  Jin]{banerjee2022online}
Siddhartha Banerjee, Vasilis Gkatzelis, Artur Gorokh, and Billy Jin.
\newblock Online nash social welfare maximization with predictions.
\newblock In \emph{Proceedings of the 33rd Annual ACM-SIAM Symposium on
  Discrete Algorithms (SODA)}, pages 1--19. SIAM, 2022.

\bibitem[Banerjee et~al.(2023)Banerjee, Hssaine, and
  Sinclair]{banerjee2023online}
Siddhartha Banerjee, Chamsi Hssaine, and Sean~R Sinclair.
\newblock Online fair allocation of perishable resources.
\newblock In \emph{Proceedings of the International Conference on Measurement
  and Modeling of Computer Systems (SIGMETRICS)}, pages 55--56, 2023.

\bibitem[Bansal and Spencer(2020)]{bansal2020line}
Nikhil Bansal and Joel~H Spencer.
\newblock On-line balancing of random inputs.
\newblock \emph{Random Structures \& Algorithms}, 57\penalty0 (4):\penalty0
  879--891, 2020.

\bibitem[Bansal et~al.(2020)Bansal, Jiang, Singla, and Sinha]{bansal2020online}
Nikhil Bansal, Haotian Jiang, Sahil Singla, and Makrand Sinha.
\newblock Online vector balancing and geometric discrepancy.
\newblock In \emph{Proceedings of the 52nd Annual ACM Symposium on Theory of
  Computing (STOC)}, pages 1139--1152, 2020.

\bibitem[Bansal et~al.(2021)Bansal, Jiang, Meka, Singla, and
  Sinha]{bansal2021online}
Nikhil Bansal, Haotian Jiang, Raghu Meka, Sahil Singla, and Makrand Sinha.
\newblock Online discrepancy minimization for stochastic arrivals.
\newblock In \emph{Proceedings of the 32nd Annual ACM-SIAM Symposium on
  Discrete Algorithms (SODA)}, pages 2842--2861, 2021.

\bibitem[Bansal et~al.(2022)Bansal, Jiang, Meka, Singla, and
  Sinha]{bansal2022prefix}
Nikhil Bansal, Haotian Jiang, Raghu Meka, Sahil Singla, and Makrand Sinha.
\newblock Prefix discrepancy, smoothed analysis, and combinatorial vector
  balancing.
\newblock In \emph{Proceedings of the 13th Innovations in Theoretical Computer
  Science Conference (ITCS)}, 2022.

\bibitem[Barman et~al.(2022)Barman, Khan, and Maiti]{barman2022universal}
Siddharth Barman, Arindam Khan, and Arnab Maiti.
\newblock Universal and tight online algorithms for generalized-mean welfare.
\newblock In \emph{Proceedings of the 36th AAAI Conference on Artificial
  Intelligence (AAAI)}, pages 4793--4800, 2022.

\bibitem[Benad{\`e} et~al.(2018)Benad{\`e}, Kazachkov, Procaccia, and
  Psomas]{benade2018make}
Gerdus Benad{\`e}, Aleksandr~M Kazachkov, Ariel~D Procaccia, and
  Christos-Alexandros Psomas.
\newblock How to make envy vanish over time.
\newblock In \emph{Proceedings of the 19th ACM Conference on Economics and
  Computation (EC)}, pages 593--610, 2018.

\bibitem[Benad{\`e} et~al.(2024{\natexlab{a}})Benad{\`e}, Halpern, Psomas, and
  Verma]{benade2024existence}
Gerdus Benad{\`e}, Daniel Halpern, Alexandros Psomas, and Paritosh Verma.
\newblock On the existence of envy-free allocations beyond additive valuations.
\newblock In \emph{Proceedings of the 25th ACM Conference on Economics and
  Computation (EC)}, 2024{\natexlab{a}}.

\bibitem[Benad{\`e} et~al.(2024{\natexlab{b}})Benad{\`e}, Kazachkov, Procaccia,
  Psomas, and Zeng]{benade2024fair}
Gerdus Benad{\`e}, Aleksandr~M Kazachkov, Ariel~D Procaccia, Alexandros Psomas,
  and David Zeng.
\newblock Fair and efficient online allocations.
\newblock \emph{Operations Research}, 72\penalty0 (4):\penalty0 1438--1452,
  2024{\natexlab{b}}.

\bibitem[Benad{\`e} et~al.(2025)Benad{\`e}, Halpern, and
  Psomas]{benade2025dynamic}
Gerdus Benad{\`e}, Daniel Halpern, and Alexandros Psomas.
\newblock Dynamic fair division with partial information.
\newblock \emph{Operations Research}, 2025.

\bibitem[Bogomolnaia et~al.(2022)Bogomolnaia, Moulin, and
  Sandomirskiy]{bogomolnaia2022fair}
Anna Bogomolnaia, Herv{\'e} Moulin, and Fedor Sandomirskiy.
\newblock On the fair division of a random object.
\newblock \emph{Management Science}, 68\penalty0 (2):\penalty0 1174--1194,
  2022.

\bibitem[Dadush et~al.(2016)Dadush, Garg, Lovett, and
  Nikolov]{dadush2016towards}
Daniel Dadush, Shashwat Garg, Shachar Lovett, and Aleksandar Nikolov.
\newblock Towards a constructive version of banaszczyk's vector balancing
  theorem.
\newblock \emph{arXiv preprint arXiv:1612.04304}, 2016.

\bibitem[Dickerson et~al.(2014)Dickerson, Goldman, Karp, Procaccia, and
  Sandholm]{dickerson2014computational}
John Dickerson, Jonathan Goldman, Jeremy Karp, Ariel Procaccia, and Tuomas
  Sandholm.
\newblock The computational rise and fall of fairness.
\newblock In \emph{Proceedings of the 28th AAAI Conference on Artificial
  Intelligence (AAAI)}, 2014.

\bibitem[Dubhashi and Panconesi(2009)]{dubhashi2009concentration}
Devdatt~P Dubhashi and Alessandro Panconesi.
\newblock \emph{Concentration of measure for the analysis of randomized
  algorithms}.
\newblock Cambridge University Press, 2009.

\bibitem[Dvoretzky et~al.(1956)Dvoretzky, Kiefer, and
  Wolfowitz]{dvoretzky1956asymptotic}
Aryeh Dvoretzky, Jack Kiefer, and Jacob Wolfowitz.
\newblock Asymptotic minimax character of the sample distribution function and
  of the classical multinomial estimator.
\newblock \emph{The Annals of Mathematical Statistics}, pages 642--669, 1956.

\bibitem[Farhadi et~al.(2019)Farhadi, Ghodsi, Hajiaghayi, Lahaie, Pennock,
  Seddighin, Seddighin, and Yami]{farhadi2019fair}
Alireza Farhadi, Mohammad Ghodsi, Mohammad~Taghi Hajiaghayi, Sebastien Lahaie,
  David Pennock, Masoud Seddighin, Saeed Seddighin, and Hadi Yami.
\newblock Fair allocation of indivisible goods to asymmetric agents.
\newblock \emph{Journal of Artificial Intelligence Research}, 64:\penalty0
  1--20, 2019.

\bibitem[Gan et~al.(2019)Gan, Suksompong, and Voudouris]{gan2019envy}
Jiarui Gan, Warut Suksompong, and Alexandros~A Voudouris.
\newblock Envy-freeness in house allocation problems.
\newblock \emph{Mathematical Social Sciences}, 101:\penalty0 104--106, 2019.

\bibitem[Gao et~al.(2021)Gao, Peysakhovich, and Kroer]{gao2021online}
Yuan Gao, Alex Peysakhovich, and Christian Kroer.
\newblock Online market equilibrium with application to fair division.
\newblock \emph{Proceedings of the 34th Annual Conference on Neural Information
  Processing Systems (NeurIPS)}, pages 27305--27318, 2021.

\bibitem[Gkatzelis et~al.(2021)Gkatzelis, Psomas, and Tan]{gkatzelis2021fair}
Vasilis Gkatzelis, Alexandros Psomas, and Xizhi Tan.
\newblock Fair and efficient online allocations with normalized valuations.
\newblock In \emph{Proceedings of the 35th AAAI Conference on Artificial
  Intelligence (AAAI)}, pages 5440--5447, 2021.

\bibitem[He et~al.(2019)He, Procaccia, Psomas, and Zeng]{he2019achieving}
Jiafan He, Ariel~D Procaccia, Alexandros Psomas, and David Zeng.
\newblock Achieving a fairer future by changing the past.
\newblock In \emph{Proceedings of the 28th International Joint Conference on
  Artificial Intelligence (IJCAI)}, pages 343--349, 2019.

\bibitem[Huang et~al.(2023)Huang, Li, Shu, and Wei]{huang2023online}
Zhiyi Huang, Minming Li, Xinkai Shu, and Tianze Wei.
\newblock Online nash welfare maximization without predictions.
\newblock In \emph{Proceedings of the 19th Conference on Web and Internet
  Economics (WINE)}, pages 402--419. Springer, 2023.

\bibitem[Igarashi et~al.(2024)Igarashi, Lackner, Nardi, and
  Novaro]{igarashi2024repeated}
Ayumi Igarashi, Martin Lackner, Oliviero Nardi, and Arianna Novaro.
\newblock Repeated fair allocation of indivisible items.
\newblock In \emph{Proceedings of the 38th AAAI Conference on Artificial
  Intelligence (AAAI)}, pages 9781--9789, 2024.

\bibitem[Kaas and Buhrman(1980)]{kaas1980mean}
Rob Kaas and Jan~M Buhrman.
\newblock Mean, median and mode in binomial distributions.
\newblock \emph{Statistica Neerlandica}, 34\penalty0 (1):\penalty0 13--18,
  1980.

\bibitem[Kash et~al.(2014)Kash, Procaccia, and Shah]{kash2014no}
Ian Kash, Ariel~D Procaccia, and Nisarg Shah.
\newblock No agent left behind: Dynamic fair division of multiple resources.
\newblock \emph{Journal of Artificial Intelligence Research}, 51:\penalty0
  579--603, 2014.

\bibitem[Kawase and Sumita(2022)]{kawase2022online}
Yasushi Kawase and Hanna Sumita.
\newblock Online max-min fair allocation.
\newblock In \emph{Proceedings of the 15th International Symposium on
  Algorithmic Game Theory (SAGT)}, pages 526--543, 2022.

\bibitem[Kulkarni et~al.(2024)Kulkarni, Reis, and
  Rothvoss]{kulkarni2024optimal}
Janardhan Kulkarni, Victor Reis, and Thomas Rothvoss.
\newblock Optimal online discrepancy minimization.
\newblock In \emph{Proceedings of the 56th Annual ACM Symposium on Theory of
  Computing (STOC)}, pages 1832--1840, 2024.

\bibitem[Kulkarni et~al.(2025)Kulkarni, Mehta, and Shahkar]{kulkarni2025online}
Pooja Kulkarni, Ruta Mehta, and Parnian Shahkar.
\newblock Online fair division: Towards ex-post constant mms guarantees.
\newblock \emph{arXiv preprint arXiv:2503.02088}, 2025.

\bibitem[Kurokawa et~al.(2016)Kurokawa, Procaccia, and Wang]{kurokawa2016can}
David Kurokawa, Ariel Procaccia, and Junxing Wang.
\newblock When can the maximin share guarantee be guaranteed?
\newblock In \emph{Proceedings of the 30th AAAI Conference on Artificial
  Intelligence (AAAI)}, volume~30, 2016.

\bibitem[Lee et~al.(2019)Lee, Kusbit, Kahng, Kim, Yuan, Chan, See, Noothigattu,
  Lee, Psomas, et~al.]{lee2019webuildai}
Min~Kyung Lee, Daniel Kusbit, Anson Kahng, Ji~Tae Kim, Xinran Yuan, Allissa
  Chan, Daniel See, Ritesh Noothigattu, Siheon Lee, Alexandros Psomas, et~al.
\newblock We{B}uild{AI}: Participatory framework for fair and efficient
  algorithmic governance.
\newblock \emph{Proceedings of the 3rd ACM Conference on Computer-Supported
  Cooperative Work and Social Computing (CSCW)}, pages 1--35, 2019.

\bibitem[Li et~al.(2018)Li, Li, and Li]{li2018dynamic}
Bo~Li, Wenyang Li, and Yingkai Li.
\newblock Dynamic fair division problem with general valuations.
\newblock In \emph{Proceedings of the 27th International Joint Conference on
  Artificial Intelligence (IJCAI)}, pages 375--381, 2018.

\bibitem[Liao et~al.(2022)Liao, Gao, and Kroer]{liao2022nonstationary}
Luofeng Liao, Yuan Gao, and Christian Kroer.
\newblock Nonstationary dual averaging and online fair allocation.
\newblock \emph{Proceedings of the 35th Annual Conference on Neural Information
  Processing Systems (NeurIPS)}, pages 37159--37172, 2022.

\bibitem[Lov{\'a}sz and Plummer(2009)]{lovasz2009matching}
L{\'a}szl{\'o} Lov{\'a}sz and Michael~D Plummer.
\newblock \emph{Matching theory}, volume 367.
\newblock American Mathematical Soc., 2009.

\bibitem[Lov{\'a}sz et~al.(1986)Lov{\'a}sz, Spencer, and
  Vesztergombi]{lovasz1986discrepancy}
L{\'a}szl{\'o} Lov{\'a}sz, Joel Spencer, and Katalin Vesztergombi.
\newblock Discrepancy of set-systems and matrices.
\newblock \emph{European Journal of Combinatorics}, 7\penalty0 (2):\penalty0
  151--160, 1986.

\bibitem[Manurangsi and Suksompong(2020)]{manurangsi2020envy}
Pasin Manurangsi and Warut Suksompong.
\newblock When do envy-free allocations exist?
\newblock \emph{SIAM Journal on Discrete Mathematics}, 34\penalty0
  (3):\penalty0 1505--1521, 2020.

\bibitem[Manurangsi and Suksompong(2021)]{manurangsi2021closing}
Pasin Manurangsi and Warut Suksompong.
\newblock Closing gaps in asymptotic fair division.
\newblock \emph{SIAM Journal on Discrete Mathematics}, 35\penalty0
  (2):\penalty0 668--706, 2021.

\bibitem[Mertzanidis et~al.(2024)Mertzanidis, Psomas, and
  Verma]{mertzanidis2024}
Marios Mertzanidis, Alexandros Psomas, and Paritosh Verma.
\newblock Automating food drop: The power of two choices for dynamic and fair
  food allocation.
\newblock In \emph{Proceedings of the 25th ACM Conference on Economics and
  Computation (EC)}, page 243, 2024.

\bibitem[Sinclair et~al.(2022)Sinclair, Jain, Banerjee, and
  Yu]{sinclair2022sequential}
Sean~R Sinclair, Gauri Jain, Siddhartha Banerjee, and Christina~Lee Yu.
\newblock Sequential fair allocation: Achieving the optimal envy-efficiency
  trade-off curve.
\newblock \emph{Operations Research}, 71\penalty0 (5):\penalty0 1689--1705,
  2022.

\bibitem[Smith and Vamanamurthy(1989)]{smith1989small}
David~J Smith and Mavina~K Vamanamurthy.
\newblock How small is a unit ball?
\newblock \emph{Mathematics Magazine}, 62\penalty0 (2):\penalty0 101--107,
  1989.

\bibitem[Spencer(1977)]{spencer1977balancing}
Joel Spencer.
\newblock Balancing games.
\newblock \emph{Journal of Combinatorial Theory, Series B}, 23\penalty0
  (1):\penalty0 68--74, 1977.

\bibitem[Spencer(1994)]{spencer1994ten}
Joel Spencer.
\newblock \emph{Ten lectures on the probabilistic method}.
\newblock SIAM, 1994.

\bibitem[Springer et~al.(2022)Springer, Hajiaghayi, Panigrahi, and
  Khani]{springer2022online}
Max Springer, MohammadTaghi Hajiaghayi, Debmalya Panigrahi, and Mohammad Khani.
\newblock Online algorithms for the santa claus problem.
\newblock In \emph{Proceedings of the 35th Annual Conference on Neural
  Information Processing Systems (NeurIPS)}, 2022.

\bibitem[Suksompong(2016)]{suksompong2016asymptotic}
Warut Suksompong.
\newblock Asymptotic existence of proportionally fair allocations.
\newblock \emph{Mathematical Social Sciences}, 81:\penalty0 62--65, 2016.

\bibitem[Vardi et~al.(2022)Vardi, Psomas, and Friedman]{vardi2022dynamic}
Shai Vardi, Alexandros Psomas, and Eric Friedman.
\newblock Dynamic fair resource division.
\newblock \emph{Mathematics of Operations Research}, 47\penalty0 (2):\penalty0
  945--968, 2022.

\bibitem[Vershynin(2018)]{vershynin2018high}
Roman Vershynin.
\newblock \emph{High-dimensional probability: An introduction with applications
  in data science}, volume~47.
\newblock Cambridge university press, 2018.

\bibitem[Yang(2023)]{yang2023fairly}
Mingwei Yang.
\newblock Fairly allocating (contiguous) dynamic indivisible items with few
  adjustments.
\newblock In \emph{Proceedings of the 22nd International Conference on
  Autonomous Agents and Multi-Agent Systems (AAMAS)}, pages 2655--2657, 2023.

\bibitem[Yang et~al.(2024)Yang, Liao, Gao, and Kroer]{yang2024online}
Zongjun Yang, Luofeng Liao, Yuan Gao, and Christian Kroer.
\newblock Online fair allocation with best-of-many-worlds guarantees.
\newblock arXiv:2408.02403, 2024.

\bibitem[Zhou et~al.(2023)Zhou, Bai, and Wu]{zhou2023multi}
Shengwei Zhou, Rufan Bai, and Xiaowei Wu.
\newblock Multi-agent online scheduling: Mms allocations for indivisible items.
\newblock In \emph{Proceedings of the 40th International Conference on Machine
  Learning (ICML)}, pages 42506--42516, 2023.

\end{thebibliography}

\newpage
\appendix


\section{Missing proofs from \Cref{sec: oblivious}}\label{app: missing from section 3}

\subsection{Proof of~\Cref{theorem:tree-subgaussianity}}

We use the following proposition from~\cite{kulkarni2024optimal}, which gives a specific symmetric convex body, whose Gaussian measure is close to $1$.

\begin{proposition}[\cite{kulkarni2024optimal}]\label{proposition:convex-body}
    For any $d, N \in \mathbb{N}$, and $\delta > 0$, let 
    $K_{\delta} \coloneqq \{ (y^{(1)}, \ldots, y^{(N)}) \in \mathbb{R}^{N d}: \norm{Y}_{\psi_2,\infty} \leq 2 + \delta$, where $Y$ picks a vector uniformly at random from the set $\{y^{(1)}, \ldots, y^{(N)}\} \}$ be a symmetric convex body. For any $\delta > 0$, there exists a constant $C_{\delta} > 0$ such that for all $d,N \in \mathbb{N}$, we have $\gamma_{Nd}(K_\delta) \geq 1 - \frac{C_{\delta}^d}{N^{1+\delta}}$.
\end{proposition}

    We will construct a tree $\calT' = (V',E')$ and then invoke~\Cref{theorem:tree-reduction} using $\calT'$ and the convex body defined in~\Cref{proposition:convex-body}. To construct $\calT'$ we start with the tree $\calT = (V,E)$ and replace every $e \in E$ by a path of $N$ edges $e^{(1)}, e^{(2)}, \ldots, e^{(N)}$, for a suitable $N \in \mathbb{N}$ to be determined later in this proof; hence, we have $|E'| = N|E|$. To every edge $e^{(i)}$ we associate a set $S_e^{(i)} = \{(\mathbf{0}, \mathbf{0}, \ldots, v, \ldots, \mathbf{0}) \in \mathbb{R}^{Nd}: v \in S_e\}$ where an element/vector in $S_e^{(i)}$ can be thought of as an $N$ by $d$ matrix, where the $i^{th}$ row is a vector from $S_e$, and all other rows are $\mathbf{0}$ vectors. This completes our construction of $\calT'$. Note that, $\mathbf{0} \in \conv(S_e)$ implies that $\mathbf{0} \in \conv(S_e^{(i)})$, for all $i \in [N]$. Additionally, by Caratheodory's theorem $\mathbf{0} \in \Ball^d_2$ can be represented as a convex combination of at most $d+1$ vectors in $S_e$, which further implies that $\mathbf{0} \in \Ball^{Nd}_2$ can also be represented as a convex combination of at most $\ell = d+1$ vectors in $S^{(i)}_e$. We also consider the symmetric convex body $K_\delta$ (as defined in~\Cref{proposition:convex-body}) with $N = \left\lceil \left((d+1)C_\delta^d|E|\right)^{1/\delta}\right\rceil$ for $\delta = 0.01$. This particular choice of $N$ gives us $\gamma_{Nd}(K_\delta) \geq 1 - \frac{C_\delta^d}{N^{1+\delta}} \geq 1 - \frac{1}{(d+1)N|E|} = 1 - \frac{1}{\ell \cdot |E'|}$. 
    
    The tree $\calT'$, along with the symmetric convex body $K_\delta$, satisfy the conditions of~\Cref{theorem:tree-reduction}, and therefore for all $e \in E$ and $i \in [N]$ there exists $s^{(i)}_e \in S^{(i)}_e$ such that for all $u \in V$, $\sum_{e \in P_u} \sum_{i=1}^{N} s^{(i)}_e = \sum_{i=1}^{N} \left( \sum_{e \in P_u} s^{(i)}_e \right) \in 11 K_\delta$. From the definition of $K_\delta$ (for $\delta = 0.01$), for $v \in K$, if we view $v$ as an $N$ by $d$ matrix, the random variable that picks a row of $v$ uniformly at random is $2.01$ subgaussian. 
    For a fixed $i$, $\sum_{e \in P_u} s^{(i)}_e$ is an element of $S_e^{(i)} \subseteq \mathbb{R}^{N d}$, i.e., an $N$ by $d$ matrix with all rows equal to the $\mathbf{0}$ vector, except row $i$. And therefore, $\sum_{i=1}^{N} \left( \sum_{e \in P_u} s^{(i)}_e \right)$ can be thought of as an $N$ by $d$ matrix whose $i^{th}$ row is exactly the $i^{th}$ row of $\sum_{e \in P_u} s^{(i)}_e$. Therefore, since $\sum_{i=1}^{N} \left( \sum_{e \in P_u} s^{(i)}_e \right) \in 11 K$ for all $u \in V$, the random variable $\sum_{e \in P_u} s^{(j)}_e$ (supported on $\mathbb{R}^{Nd}$) where $j \sim \mathcal{U}([N])$ is $\left(11 \cdot 2.01 \right) = 22.11$-subgaussian. And, since all but the $i^{th}$ row of $s^{(i)}_e$ are equal to the $\mathbf{0}$ vector, we also have that the distribution $\calD$ that samples $j \sim \mathcal{U}([N])$ and then outputs the $j$-th row of $s^{(j)}_e$, for all $e \in E$, a distribution supported on $\bigtimes_{e \in E} S_e$, is also $22.11$-subgaussian. \qed

\subsection{Proof of~\Cref{lemma:weighted-vector-balancing}}
For any convex body $K \in \mathbb{R}^d$ and vector $v \in \mathbb{R}^d$, we define the \emph{gauge function} $\norm{v}_K \coloneqq \inf\{x \mid v \in x K\}$. To prove \Cref{lemma:weighted-vector-balancing}, we will use the following tail bound version of Talagrand's comparison inequality (Chapter 8.6 ~\cite{vershynin2018high}). See~\cite{kulkarni2024optimal} for a similar application of this inequality.
\begin{lemma}\label{lemma:talagrand_comparison_inequality}
    Let $K \in \mathbb{R}^n$ be a symmetric convex body and $X \in \mathbb{R}^n$ be a random vector that is $O(1)$-subgaussian. Then, with probability at least $1-\delta$
    $$\norm{X}_K \lesssim \E_{g \sim N(\mathbf{0}, I_n)}[\norm{g}_K] + \sqrt{\log(1/\delta)}\cdot \frac{1}{\mathrm{inradius}(K)},$$
    for any $\delta \in (0,1/2]$, where $\mathrm{inradius}(K)$ is the largest $r > 0$ such that $rB_2^n \subseteq K$.
\end{lemma}

We will now restate and prove~\Cref{lemma:weighted-vector-balancing}.

\InftyNormBound*
\begin{proof}
    Consider the algorithm of~\Cref{theorem:subgauss-algo} where the set of vectors at time $t$ is $S_t = \{(1-\alpha)v_t, -\alpha v_t\}$. Let $s_t \in S_t$ be the vector chosen by this algorithm. If $s_t = (1-\alpha) v_t$, set $w_t = 1-\alpha$; otherwise, set $w_t = -\alpha$. Therefore, we have $\sum_{i=1}^t w_iv_i = \sum_{i=1}^t s_i$, where from~\Cref{theorem:subgauss-algo} we know that $\sum_{i=1}^t s_i$ is $O(1)$-subgaussian. Using \Cref{lemma:talagrand_comparison_inequality} we get that with probability at least $1 - \delta$ for some $\delta \in (0,1/2]$ we have,

    \begin{align*}
        \norm{\sum_{i=1}^t w_iv_i}_\infty = \norm{\sum_{i=1}^t s_i}_\infty = \norm{\sum_{i=1}^t s_i}_{B^d_\infty} & \lesssim \E_{g \sim N(\mathbf{0}, I_n)}[\norm{g}_{B^d_\infty}] + \sqrt{\log(1/\delta)}\cdot \frac{1}{\mathrm{inradius}(B^d_\infty)}\\
        & = \sqrt{\log(4d)} \E_{g \sim N(\mathbf{0}, I_n)}[\norm{g}_{\sqrt{\log(4d)} \cdot B^d_\infty}] + \sqrt{\log(1/\delta)} \numberthis \label{equation:talagrand_lemma_eq_1},
    \end{align*}
    where the final equality uses the fact that $\mathrm{inradius}(B^d_\infty) = 1$ and for any $p > 0$ we have $\norm{v}_K = p \norm{v}_{pK}$. To further simplify this expression, we will prove the following claim.

    \begin{claim}\label{claim:infinity-norm-bound-lemma-claim}
        $\E_{g \sim N(\mathbf{0}, I_n)}[\norm{g}_{\sqrt{\log(4d)} \cdot B^d_\infty}] = O(1).$
    \end{claim}
    \begin{proof}
        To prove this, we will first show that $\gamma_d(\sqrt{\log(4d)} \cdot B^d_\infty) \geq 1/2$. Then applying \cite[Lemma 26]{dadush2016towards}, which states that $\E_{g \sim N(\mathbf{0}, I_d)}[\norm{g}_K] \leq 1.5$ if $\gamma_d(K) \geq 1/2$, will give us $\E_{g \sim N(\mathbf{0}, I_n)}[\norm{g}_{\sqrt{\log(4d)} \cdot B^d_\infty}] \leq 1.5 = O(1)$, establishing the claim. Hence, we consider the following,

        \begin{align*}
            \gamma_d(\sqrt{\log(4d)} \cdot B^d_\infty) & = \mathop{\mathbb{P}}\limits_{g \sim N(\mathbf{0},I_d)}[ g \in \sqrt{\log(4d)} \cdot B^d_\infty]\\
            & = \mathop{\mathbb{P}}\limits_{g_1, g_2,\ldots, g_d \sim N(0,1)}\left[\cap_{j=1}^d \left(g_j \leq \sqrt{\log(4d)}\right)\right]\\
            & = 1 - \mathop{\mathbb{P}}\limits_{g_1, g_2,\ldots, g_d \sim N(0,1)}\left[\cup_{j=1}^d \left(g_j > \sqrt{\log(4d)}\right)\right]\\
            & \geq 1 - d \cdot \mathop{\mathbb{P}}\limits_{g_1 \sim N(0,1)}\left[g_j > \sqrt{\log(4d)}\right]\\
            & \geq 1 - d\cdot 2 e^{-(\sqrt{\log(4d)})^2} = 1/2,
        \end{align*}
        where the final inequality follows from the fact that for a standard Gaussian with variance $1$, we have $\mathbb{P}_{g \sim N(0,1)}[g > k] \leq 2 e^{-k^2}$. The claim stands proved.
    \end{proof}
    Now combining \Cref{equation:talagrand_lemma_eq_1} with \Cref{claim:infinity-norm-bound-lemma-claim}, gives us that with probability at least $1 - \delta$, we have $\norm{\sum_{i=1}^t w_iv_i}_\infty \lesssim \sqrt{\log(4d)} + \sqrt{\log(1/\delta})$. This implies that, with probability at most $\delta / T$, we have $\norm{\sum_{i=1}^t w_iv_i}_\infty \gtrsim \sqrt{\log(4d)} + \sqrt{\log(T/\delta})$. Taking union bound over all time steps $t \in [T]$, we get that with probability at least $1 - \delta$, we have $\max_{t=1}^T \norm{\sum_{i=1}^t w_iv_i}_\infty \lesssim \sqrt{\log(4d)} + \sqrt{\log(T/\delta})$. Focusing on the dependency on $T$ and $\delta$, we get $\max_{t=1}^T \norm{\sum_{i=1}^t w_iv_i}_\infty \lesssim \sqrt{\log(T)} + \sqrt{\log(1/\delta)}$.
\end{proof}


\subsection{Proof of Theorem~\ref{theorem:adaptive-lb}}\label{app:proof from OR paper}

In this section, we show that, for any $n \geq 2$, $r < 1$ and $T \geq 1$, in the online envy minimization problem, there exists a set of instances $S_T$, with $|S_T|  \leq 2^T$, such that, for any deterministic online fair division algorithm $\mathcal{A}$, there exists an instance $I \in S_T$, such that running algorithm $\mathcal{A}$ on the sequence of items $1,2 \ldots, T$ described by $I$ results in $\env^T \in \Omega((T/n)^{r/2})$. We first prove the bound for $n=2$, followed by the case of an arbitrary number of agents.


\begin{lemma}\label{lem:LBn=2}
For $n = 2$ and any $r < 1$, there exists a set of instances $S_T$, with $|S_T|  \leq 2^T$, such that, for any online fair division algorithm $\mathcal{A}$, there exists an instance $I \in S_T$, such that running algorithm $\mathcal{A}$ on the sequence of items $1,2 \ldots, T$ described by $I$ results in $\env^T \in \Omega(T^{r/2})$.
\end{lemma}

\begin{proof}
We will describe a strategy for the adaptive adversary. The adversary will generate an instance overtime, as the algorithm makes its choices. The set of instances $S_T$ is the set of all possible instances; $|S_T|  \leq 2^T$ since the algorithm makes a binary choice at each step.

Label the agents $L$ and $R$, and let $\{v_0 = 1,v_1,v_2,\ldots\}$ be a decreasing sequence of values (specified later) satisfying $v_{d} - v_{d+1} < v_{d'} - v_{d'+1}$ for all $d' < d$. The adversary keeps track of the \emph{state} of the game, and the current state defines its strategy for choosing the agents' valuations for the next item. The lower bound follows from the adversary strategy illustrated in Figure~\ref{fig:LB}. Start in state $0$, which we will also refer to as $L_0$ and $R_0$, where the adversary sets the value of the arriving item as $(1,1)$. To the left of state $0$ are states labeled $L_1,L_2,\ldots$; when in state $L_d$, the next item that arrives has value $(1,v_d)$. 
To the right of state $0$ are states labeled $R_1,R_2,\ldots$; when in state $R_d$ the next item  arrives with value $(v_d,1)$. Whenever the algorithm allocates an item to agent $L$ (resp.\ $R$), which we will refer to as making an $L$ (resp.\ $R$) step, the adversary moves one state to the left (resp.\ right).

\begin{figure}[b]
        \centering
	\begin{tikzpicture}[->,>=stealth,auto,node distance=2cm,
    thick,every node/.style={circle,draw,font=\sffamily,minimum size=3em,inner sep=1}]

		\node [label={[label distance=-0.75em]north:{\textbf{0}}}] (0) {$(1,1)$};
		\node [right of=0,label={[label distance=-0.75em]north:{$\mathbf{R_1}$}}] (1) {$(v_1,1)$};
		\node [right of=1,label={[label distance=-0.75em]north:{$\mathbf{R_2}$}}] (2) {$(v_2,1)$};
		\node [right of=2,draw=none] (3) {$\cdots$};  
		\node [left of=0,label={[label distance=-0.75em]north:{$\mathbf{L_1}$}}] (-1) {$(1,v_1)$};
		\node [left of=-1,label={[label distance=-0.75em]north:{$\mathbf{L_2}$}}] (-2) {$(1,v_2)$};
		\node [left of=-2,draw=none] (-3) {$\cdots$};  
		\path 
		(-3) edge [bend left] node [draw=none,label={[align=center,above=-4em]\small $(1,-v_3)$}] {} (-2)
		(-2) edge [bend left] node [draw=none,label={[align=center,above=-4em]\small $(1,-v_2)$}] {} (-1)
		(-1) edge [bend left] node [draw=none,label={[align=center,above=-4em]\small $(1,-v_1)$}] {} (0)
		(0) edge [bend left] node [draw=none,label={[align=center,above=-3.81em]\small $(1,-1)$}] {} (1)
		(1) edge [bend left] node [draw=none,label={[align=center,above=-4em]\small $(v_1,-1)$}] {}  (2)
		(2) edge [bend left] node [draw=none,label={[align=center,above=-4em]\small $(v_2,-1)$}] {}  (3);
		\path 
		(3) edge [bend left] node [draw=none,label={[align=center,below=-1em]\small $(-v_3,1)$}] {} (2)
		(2) edge [bend left] node [draw=none,label={[align=center,below=-1em]\small $(-v_2,1)$}] {} (1)
		(1) edge [bend left] node [draw=none,label={[align=center,below=-1em]\small $(-v_1,1)$}] {} (0)
		(0) edge [bend left] node [draw=none,label={[align=center,below=-0.81em]\small $(-1,1)$}] {} (-1)
		(-1) edge [bend left] node [draw=none,label={[align=center,below=-1em]\small $(-1,v_1)$}] {} (-2)
		(-2) edge [bend left] node [draw=none,label={[align=center,below=-1em]\small $(-1,v_2)$}] {} (-3);
		\end{tikzpicture}
	\caption{Adversary strategy for the two-agent lower bound. In state $L_d$, an item valued $(1,v_d)$ arrives, while in state $R_d$, an item valued $(v_d,1)$ arrives. The arrows indicate whether agent $L$ or agent $R$ is given the item in each state. The arrows are labeled by the amount envy changes after that item is allocated.}\label{fig:LB}
\end{figure}
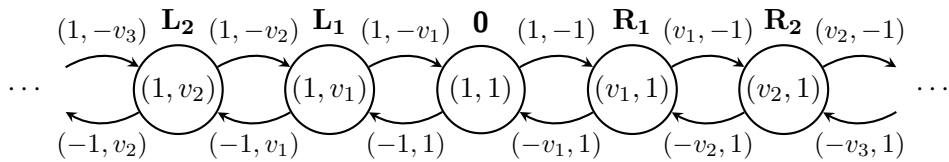

We construct the optimal allocation algorithm against this adversary, and show that for this algorithm the envy at some time step $t \in [T]$ will be at least $\Omega(T^{r/2})$ for the given $r<1$.
This immediately implies Lemma~\ref{lem:LBn=2}: if the envy is sufficiently large at any time step $t$ the adversary can guarantee the same envy at time $T$ by making all future items valued at zero by both agents.

The intuition for the adversary strategy we have defined is that it forces the algorithm to avoid entering state $L_d$ or $R_d$ for high $d$, as otherwise the envy of some agent will grow to $v_0 + v_1 + \cdots + v_d$, which will be large by our choice of $\{v_d\}$. At the same time, if an $L$ step is taken at state $L_d$, followed by a later return to state $L_d$, the envy of $R$ increases by at least $v_d-v_{d+1}$; we choose $\{v_d\}$ so that this increase in envy is large enough to ensure that any algorithm which spends too many time steps close to state $0$ incurs large envy.

By the pigeonhole principle, either the states to the left or to the right of state $0$ are visited for at least half the time.
Assume, without loss of generality, that our optimal algorithm spends time $T' = \left\lceil{T/2}\right\rceil$ in the ``left'' states ($L_0, L_1,\ldots$), and that $T'$ is even.
We prove that the envy of agent $R$ grows large at some time step $t$.
We ignore any time the algorithm spends in the states $R_d$, $d \ge 1$.
To see why this is without loss of generality, consider first a cycle spent in the right states that starts at $R_0$ with an item allocated to $R$ and eventually returns to $R_0$.
In such a cycle, an equal number of items are allocated to both agents.
All of these items have value $1$ to agent $R$, yielding a net effect of $0$ on agent $R$'s envy. ({We ignore agent $L$ completely, as our analysis is of the envy of agent $R$.}) The other case is when the algorithm starts at $R_0$ but does not return to $R_0$. This scenario can only occur once, which means that the algorithm has already taken $T'$ steps on the left side; the allocation of these items does not affect our proof.

Let  $0 \leq K \leq T'/2$ be an integer and denote by $\OPT(K)$ the set of envy-minimizing allocation algorithms that spend the $T'$ steps in states $L_0, \ldots, L_{K}$ (and reach $L_K$). Note that the algorithm aims to minimize the maximum envy at any point in its execution. 
Let $\mathcal{A}^*(K)$ be the following algorithm, starting at $L_0$:
Allocate the first $K$ items to agent $L$, thus arriving at state $L_{K}$.
Alternate between allocating to agents $R$ and $L$ for the next $T'-2K$ items, thereby alternating between states $L_{K-1}$ and $L_{K}$.
Allocate the remaining $K$ items to agent $R$. Our first result is that $\mathcal{A}^*(K)$ belongs to $\OPT(K)$.

\begin{lemma}\label{lem:lb:A*opt}
	$\mathcal{A}^*(K) \in \OPT(K)$.
\end{lemma}
 

We analyze the envy of $\mathcal{A}^*(K)$ as a function of $K$ before optimizing $K$.
Agent $R$'s maximum envy is realized at step $T' - K$, right before the sequence of $R$ moves. $\env^{T'-K}$ has two terms: the envy accumulated to reach state $L_{K}$, and the envy from alternating $R$ and $L$ moves between states $L_{K}$ and $L_{K-1}$, so 
%
\begin{align}
\env^{T'-K} 
= \sum_{d=0}^{K-1} v_d + \frac{T'-2K}{2} \cdot \left( v_{K-1}-v_{K} \right). \label{eq:adaptive1}
\end{align}
Given $r<1$, define $v_d \defeq (d+1)^r - d^r$.  Notice that $\sum_{d=0}^{K-1} v_d = K^r$. 
\edit{When  $K\geq \sqrt{T'/2}$ it follows that  $\sum_{d=0}^{K-1} v_d \geq (T'/2)^{r/2} \in \Omega(T^{r/2})$, which is what we set out to prove. We limit the rest of the analysis to the case where  $K\leq \sqrt{T'/2}$. }


\begin{lemma}
	\label{lem:lb:bound-diff}
	\edit{Let $K\leq \sqrt{T'/2}$ and define $v_d \defeq (d+1)^r - d^r$ for $r<1$.} Then $v_{K-1} - v_{K} \geq r (1-r) K^{r-2}$.
\end{lemma}

Applying Lemma~\ref{lem:lb:bound-diff}  to \eqref{eq:adaptive1} and distributing terms yields
\begin{align}
\env^{T'-K} \geq K^r - r (1-r) K^{r-1} + \frac{T'}{2} r (1-r) K^{r-2} \geq \frac{1}{2} \left( K^r + {T'} r (1-r) K^{r-2} \right), \label{eq:adaptive2}
\end{align}
where the second inequality uses the fact that $r(1-r) \le 1/4 < 1/2$ and assumes $K > 1$ (otherwise the envy would be linear in $T'$).
To optimize $K$, noting that the second derivative of the above bound is positive for $K\leq \sqrt{T'/2}$, we find the critical point:
\begin{align*}
\frac{\partial}{\partial K} \left(K^r + {T'} r (1-r) K^{r-2}\right) = r K^{r-1} - T'r(1-r)(2-r) K^{r-3} = 0  \implies  K = \sqrt{T'(1-r)(2-r)}.
\end{align*}
Defining $C_1 \defeq \sqrt{(1-r)(2-r)}$ and  substitute  into \eqref{eq:adaptive2} to obtain
\begin{align}
	\env^{T'-K} 
	\ge \frac{1}{2} \left(C_1^r (T')^{r/2} + T' r (1-r) C_1^{r-2} (T')^{r/2 -1} \right)
	\in \Omega(T^{r/2}),
\end{align}
completing the proof.
\end{proof}

The extension to $n$ agents follows from the same set of instances for agents $L$, $R$, letting all other agents value every item at zero. \qed

\section{Missing Proofs from Section~\ref{sec: iid}}\label{app:missing proofs from iid}

\subsection{Proof of \Cref{lem:no-cycle}}

    We first claim that a sufficient condition for avoiding cycles at time step $t$ is that $v_j(A^t_i)   \le v_i(A^t_i)  + c$ for all agents $i \ne j$. Indeed, consider a cycle of agents $i_1, \ldots, i_k, i_{k+1}=i_1$. Then, we have \begin{align*}
        \sum_{j = 1}^k \envy_{i_{j}, i_{j + 1}}  &= \sum_{j = 1}^k v_{i_j}(A_{i_{j + 1}}) - v_{i_j}(A_{i_j})\\
        &= \sum_{j = 1}^k v_{i_j}(A_{i_{j + 1}}) - \sum_{j = 1}^k v_{i_j}(A_{i_j})\\
        &= \sum_{j = 1}^k v_{i_j}(A_{i_{j + 1}}) - \sum_{j = 1}^k v_{i_{j+1}}(A_{i_{j+1}})\\
        &= \sum_{j = 1}^k v_{i_j}(A_{i_{j + 1}}) - v_{i_{j + 1}}(A_{i_{j + 1}})\\
        &\le k \cdot c.
    \end{align*}
    Thus, for at least one pair, $\envy^t_{i_j, i_{j + 1}} \le c$, preventing a cycle from forming.

    Now, fix two agents $i \ne j$. We aim to show that
    \[
        \Pr[\forall t \ge T^{(1)}, \, v_i(A^t_i) -  v_j(A^t_i) < - c] \le 4 \cdot \left(\frac{8 e n^2 \log T}{\sqrt{T}} \right)^{c + 1}. 
    \]
    Applying a union bound over the $n(n - 1)$ pairs of $i$ and $j$ yields the lemma statement.

    Without loss of generality, we relabel $i = 1$ and $j = 2$. Define $Z_t := v_1(A^t_1) - v_2(A^t_1)$ as the difference in values at time $t$. Our goal is to show that $Z_t \ge -c$ for all $t \ge T^{(1)}$ with high probability.

    We decompose $A^t_1$ into the portion received during each phase:
    \[
        v_1(A^t_1) - v_2(A^t_1) = v_1(A^t_1 \cap G^{welf} ) - v_2(A^t_1 \cap G^{welf} ) + v_1(A^t_1 \setminus G^{welf}) - v_2(A^t_1 \setminus G^{welf}). 
    \]
    Since $t \ge T^{(1)}$, we can express 
    \[
         v_1(A^t_1 \cap G^{welf} ) - v_2(A^t_1 \cap G^{welf} ) = \sum_{j = 1}^{T^{(1)}} \left(V^{g^{welf}_j}_1 - V^{g^{welf}_j}_2 \right)  \cdot \mathbb{I}[g^{welf}_j \in A^t_1].
    \]
    Define \[X_j = \left(V^{g^{welf}_j}_2 - V^{g^{welf}_j}_1 \right)  \cdot \mathbb{I}[g^{welf}_j \in A^t_1].\]
    Now, consider $A^t_1 \setminus G^{welf}$. By \Cref{lem:phase-2-items}, we have that \[|A^t_1 \setminus G^{welf}| \le (n - 1) \cdot \ceil{\log T \sqrt{T}}.\] 
    Thus, we have $A^t_1 \setminus G^{welf} \subseteq \{g^1_1, \ldots, g^1_{(n - 1) \cdot \ceil{\log T \sqrt{T}}} \}$ (recall our way of sampling an instance). Furthermore, we can lower-bound 
    \[v_1(A^t_1 \setminus G^{welf}) - v_2(A^t_1 \setminus G^{welf}) \ge -  \sum_{j = 1}^{(n - 1) \cdot \ceil{\log T \sqrt{T}}} (V^{g^i_j}_2 - V^{g^i_j}_1)_+\] where $(s)_+ = \max(s, 0)$. That is, in this lower bound we only count items where agent $2$ had a higher value than agent $1$. Let $Y_j = (V^{g^i_j}_1 - V^{g^i_j}_2)_+$.
    
    Since the above bounds are independent of $t$, it suffices to show that $$\sum_{j = 1}^{T^{(1)}} X_j - \sum_{j = 1}^{(n - 1) \ceil{\log T \sqrt{T}}} Y_j \ge -c$$  with high probability, ensuring the bound holds for all remaining $t > T^{(1)}$.

    If $X_j$ stochastically dominated $Y_j$, then a straightforward application of~\Cref{lem:concentration} would imply the statement. Unfortunately, this is not the case. So, instead, we define the following random variables, $X'_j$, that are large sums of $X_j$ random variables, and such that the stochastic dominance we want is true. For each $j \le \floor{T^{(1)}/(2n)}$, let $X'_j = \sum_{j' = (j - 1) \cdot 2n + 1}^{j \cdot 2n} X_j$, i.e., each $X'_j$ is the sum of a distinct set of $2n$ $X_j$s. We will show that  $$\sum_{j = 1}^{\floor{T^{(1)}/(2n)}} X'_j - \sum_{j = 1}^{(n - 1) \ceil{\log T \sqrt{T}}} Y_j \ge -c$$  with high probability.
    Specifically, we will show that each $X'_j$ first-order stochastically dominates each $Y_j$. Then, by applying~\Cref{lem:concentration} we have 
    \[
    \Pr \left[ \sum_{j = 1}^{\floor{T^{(1)}/(2n)}} X'_j - \sum_{j = 1}^{(n - 1) \ceil{\log T \sqrt{T}}} Y_j < -c \right] \leq 4 \cdot \left(\frac{2 e L}{K} \right)^{c + 1},
    \]
    where $K = \floor{\frac{T - \frac{n(n-1)}{2}\ceil{\log T \sqrt{T}}}{2n}} + (n - 1) \ceil{\log T \sqrt{T}} \geq \frac{T}{2n}$ (as long as $\ceil{\log T \sqrt{T}}(\frac{3n}{4} - \frac{3}{4}) \geq 1$, which holds for $T \geq 4$), and $L = (n - 1) \ceil{\log T \sqrt{T}} \leq 2n \log T \sqrt{T}$. Note that, indeed, $\frac{K}{L} \geq \frac{\sqrt{T}}{4 n^2 \log T } \geq 4e$, for $T \in \Omega(n^6)$.
    So, overall:
    \[
    \Pr \left[ \sum_{j = 1}^{\floor{T^{(1)}/(2n)}} X'_j - \sum_{j = 1}^{(n - 1) \ceil{\log T \sqrt{T}}} Y_j < -c \right] \leq 4 \cdot \left(\frac{8 e n^2 \log T}{\sqrt{T}} \right)^{c + 1}.
    \]

    To analyze this probability, we characterize the distributions of each $X_j$, $X'_j$, and $Y_j$. We will use $X$, $X'$, and $Y$ as random variables with the same distribution as each $X_j$, $X'_j$, and $Y_j$, respectively. Recall, a way to sample $X$ is to sample $n$ values $V_1, \ldots, V_n \stackrel{\text{i.i.d.}}{\sim} \mathcal{D}$, and if $V_1$ is the largest (breaking ties randomly), set $X = V_1 - V_2$, otherwise set $ X = 0$. To define $X'$ we sum $2n$ draws from $X$. Finally, a way to sample $Y$ is to sample $n$ values $V_1, \ldots, V_n \stackrel{\text{i.i.d.}}{\sim} \mathcal{D}$ and if $V_1 \ge V_2$ (breaking ties randomly), set $Y = V_1 - V_2$, otherwise set $Y = 0$.

    All of these can be understood using distributions induced by the difference of order statistics. More formally, let $\mathcal{D}^{(k) - (\ell)}$ be the distribution obtained by drawing $V_1, \ldots, V_n \stackrel{\text{i.i.d.}}{\sim} \mathcal{D}$, sorting them as $V^{(1)} \le \cdots \le V^{(n)}$, and returning $V^{(k)} - V^{(\ell)}$.
    \begin{enumerate}
        \item  The distribution of $X$ corresponds to selecting distinct indices $i_1, i_2$ uniformly from $[n]$, and if $i_1 = n$, sampling from $\mathcal{D}^{(n) - (i_2)}$; otherwise, outputting $0$. 
        \item The distribution of $Y$ corresponds to selecting distinct indices $i_1, i_2$ uniformly from $[n]$, and if $i_1 > i_2$, sample from $\mathcal{D}^{(i_1) - (i_2)}$; otherwise, outputting $0$. 
    \end{enumerate}

The distribution of $X'$ can be described as follows:
    \begin{enumerate}
        \item Sampling $2n$ pairs $(i^j_1, i^j_2)_{j = 1, \ldots, 2n}$, where $i^j_1 \neq i^j_2$ for all $j$, from $[n]$, uniformly at random.
        \item For each pair where $i^j_1 = n$, draw a value from $\calD^{(n) - (i^j_2)}$
        \item Output the sum of these values (or output $0$ if no $i^j_1 = n$).
    \end{enumerate} 

    It remains to show that $X'$ first-order stochastically dominates  $Y$. To this end, note that if $i'_1 \ge i_1$ and $i'_2 \le i_2$, then $\mathcal{D}^{(i'_1) - (i'_2)}$ first-order stochastically dominates $\mathcal{D}^{(i_1) - (i_2)}$. We now present a sequence of distributions, each distribution stochastically dominating the previous distribution in the sequence, beginning with $Y$ and ending with $X'$.

    \paragraph{Distribution 1.} We begin with the distribution of $Y$. Recall that this can be described as: draw a distinct pair $i_1, i_2 \in [n]$. If $i_1 > i_2$, sample from $\calD^{(i_1) - (i_2)}$; otherwise, output $0$. Since each ordering of $i_1$ and $i_2$ are equally likely, an equivalent description is draw a distinct pair $i_1, i_2 \in [n]$. With probability $1/2$, output $0$, and otherwise, sample from $\calD^{(\max(i_1, i_2)) - (\min(i_1, i_2))}$.

    \paragraph{Distribution 2.} With probability $1/2$, output $0$. Otherwise, sample a distinct pair $i_1, i_2$ from $[n]$ and output $\mathcal{D}^{(n) - (\min(i_2, i'_2))}$. This stochastically dominates distribution $1$ because $n \ge \max(i_1, i_2)$.

    \paragraph{Distribution 3.} With probability $1/2$, output $0$. Otherwise, sample two (not necessarily distinct) values $i_2, i_2'$ uniformly from $[n-1]$ and output $\mathcal{D}^{(n) - (\min(i_2, i'_2))}$. To show that this stochastically dominates Distribution 2,  observe that taking the minimum of a distinct pair from $[n]$ stochastically dominates taking the minimum of a (not necessarily distinct) pair from $[n - 1]$, so this only increases the likelihood of drawing from a ``better'' distribution. Specifically, the probability that $\min(i_1, i_2) \ge k$, where $i_1, i_2$ are a distinct pair from $[n]$,  is $\frac{n - k + 1}{n} \cdot \frac{n - k}{n - 1}$. The probability that $\min(i_1, i_2) \ge k$, where $i_1, i_2$ are a (possibly non-distinct) pair from $[n-1]$, is $\left(\frac{n - k}{n - 1}\right)^2$. Since $\frac{n - k + 1}{n} \, \frac{n - k}{n - 1} > \left( \frac{n - k}{n - 1} \right)^2$, the former probability is greater. As this holds for all $k$, stochastic dominance is implied.

    \paragraph{Distribution 4.} Draw $2n$ pairs $(i^j_1, i^j_2)_{j = 1, \ldots, 2n}$, where numbers in a pair are distinct, and $i^j_{\ell}$ is drawn from $[n]$. If exactly zero or exactly one pairs have $i^j_1 = n$, output $0$. Otherwise, let  $(i_1, i_2)$ and $(i'_1, i'_2)$ be the first two pairs with $i_1 = i'_1 = n$, and draw from $\mathcal{D}^{(n) - (\min(i_2, i'_2))}$. Note that conditioned on $i_1 = n$, $i_2$ is just a uniform draw from $[n - 1]$. So to establish stochastic dominance over Distribution 3, we simply need to show that the probability of at least two pairs with $i^j_1 = n$ is at least $1/2$. The number of such pairs follows a $\text{Bin}(2n, 1/n)$ distribution, which has mean $2n/n = 2$. Furthermore, it is known that the median is at least the floor of the mean~\cite{kaas1980mean}, thus, the probability of having at least two pairs is at least $1/2$, as needed.

     \paragraph{Distribution 5.} Next, we consider a distribution that only keeps the ``best'' pair.  That is, we draw $2n$ pairs  $(i^j_1, i^j_2)_{j = 1, \ldots, 2n}$ as in Distribution 4, and among those where $i^j_1 = n$, select the minimal $i^j_2$, and output $\mathcal{D}^{(n) - (i^j_2)}$ (or $0$ if no such pair exists). This stochastically dominates Distribution 4 because, in the cases when there are at least two pairs with $i^j_1 = n$, we are taking the minimum over even more values; and we are now potentially achieving a positive value even when there is only one pair with $i^j_1 = n$.

     \paragraph{Distribution 6.} Draw $2n$ pairs $(i^j_1, i^j_2)_{j = 1, \ldots, 2n}$ as in Distribution 5, and for each pair where $i^j_1 = n$, draw a value from $\calD^{(n) - (i^j_2)}$, and output the sum of these values (or output  $0$ if no $i^j_1 = n$). Note that we are only including more pairs than Distribution 5, so this stochastically dominates it. Furthermore, this is exactly the distribution of $X'$.  

     This completes the proof. \qed

\subsection{Proof of \Cref{lem:concentration}}

	A key observation is that, since the $Y_i$'s are i.i.d.\@ draws, by symmetry, the probability of the event we care about is equal to
	\[
		\Pr\left[\sum_{i \notin S} Y_{i} - \sum_{i  \in S} Y_{i} < -c\right],
	\]
	where $S \subseteq [K]$ is a randomly sampled set of indices of size $L$. Condition on the set of draws $Y_1, \ldots, Y_K$ having arbitrary values $y_1, \ldots, y_K$, and consider the randomness over the set $S$. By symmetry, it is without loss of generality that $y_1 \ge \cdots \ge y_K$.

    Now, by \Cref{lem:halls}, a sufficient condition for $\sum_{i \notin S} y_i + \sum_{i \in S} y_i \ge -c$ is that for all $j \in S$, $$|\{j' \in S \mid y_{j'} \ge y_j\}| \le |\{j' \notin S\mid y_{j'} \ge y_j\}| + c.$$
    A sufficient condition for this is that for all $j \le K$, $$|[j] \cap S| \le |[j] \setminus S| + c,$$ i.e., in any prefix there are the number of indices in $S$ never exceeds those outside of $S$ by more than $c$. Finally, observe that $|[j] \setminus S| = j - |[j] \cap S|$, we can again reformulate this as for all $j \le K$,
    \[
        |[j] \cap S| \le \frac{j + c}{2}.
    \]
	
	Let $\mathcal{E}^j$ be the event that $|[j] \cap S| > \frac{j + c}{2}.$ We will upper bound each $\Pr[\mathcal{E}^j]$ and then union bound  over all $j$. 
	
	Let $Z_i := \mathbb{I}[i \in S]$, so $|[j] \cap S| = \sum_{i = 1}^j Z_i$. 
    The $Z_i$s are not independent, but they \emph{are} negatively associated, and thus, traditional Chernoff bounds apply~\cite{dubhashi2009concentration}. We have  that $\mathbb{E}[Z_i] = \frac{L}{K}$ for each $i$. For $j \le c$, note that $\Pr[\mathcal{E}^j] = 0$ because $\sum_{i=1}^j Z_i \le j \le \frac{j + c}{2}$. For $j \ge c + 1$, note that $\mathbb{E}[\sum_{i = 1}^j Z_i]$ is precisely $\mu := \frac{j \cdot L}{K}$. We would like to upper bound the probability that $\sum_{i = 1}^j Z_i$ exceeds $\frac{j + c}{2}$. Since the $Z_i$s are integral, it suffices to bound the probability that  $\sum_{i = 1}^j Z_i$ exceeds $\ceil*{\frac{j + c}{2}}$. Let $\delta$ be such that  $\ceil*{\frac{j + c}{2}} = (1 + \delta) \mu$. Note that since $L/K \le 2$, $\mu \le j/2$, and therefore, $\delta > 0$. Furthermore, $$1 + \delta = \frac{j + c}{2 \mu} \ge \frac{j}{2\mu} =  \frac{K}{2 L}.$$ Using the Chernoff bound, we have:
	\[
		\Pr[X \ge (1 + \delta) \mu] \le \left(\frac{e^\delta}{(1 + \delta)^{1 + \delta}} \right)^\mu \leq  \left(\frac{e}{(1 + \delta)} \right)^{(1 + \delta)\mu} \le  \left(\frac{2 e L}{K} \right)^{\ceil*{\frac{j + c}{2}}}.
	\]
	
	Applying a union bound over all $j$, we have that
\begin{align*}
    \Pr\left[\bigcup_{j \le K} \mathcal{E}^j \right]
    &\le \sum_{j = 0}^K \Pr[\mathcal{E}^j]\\
    &= \sum_{j = c + 1}^K \Pr[\mathcal{E}^j]\\
    &\le \sum_{j = c + 1}^K \left(\frac{2 e L}{K} \right)^{\ceil*{\frac{j + c}{2}}}\\
    &\le \sum_{j = c + 1}^\infty \left(\frac{2 e L}{K} \right)^{\ceil*{\frac{j + c}{2}}}\\
    &\le \sum_{j = 0}^\infty \left(\frac{2 e L}{K} \right)^{\ceil*{\frac{j + 2c + 1}{2}}}\\
    &\le \sum_{j = 0}^\infty \left(\frac{2 e L}{K} \right)^{\ceil*{\frac{2j + 2c + 1}{2}}} + \left(\frac{2 e L}{K} \right)^{\ceil*{\frac{2j + 1 + 2c + 1}{2}}}\\
    &= \sum_{j = 0}^\infty 2 \left(\frac{2 e L}{K} \right)^{j + c + 1}\\
    &= 2 \left(\frac{2 e L}{K} \right)^{c + 1} \cdot \sum_{j = 0}^\infty \left(\frac{2 e L}{K} \right)^{j}\\
    &= 2 \left(\frac{2 e L}{K} \right)^{c + 1} \cdot \frac{1}{1 - \frac{2eL}{K}}\\
    &\le 4  \left(\frac{2 e L}{K} \right)^{c + 1},
\end{align*}
where the last two transitions use the fact that $\frac{2eL}{K} \le 1/2$. \qed

\subsection{Proof of \Cref{lem:phase-2-items}}
We will prove this by induction on the time steps. Note that at $T^{(1)}$, no phase 2 items have been given out, so $w^t_{i} = 0$ for all $i$, satisfying the lemma statement. Now suppose this is true at some fixed time $t$, and suppose the current sorted vector is $(w^t_{i_1}, \ldots, w^t_{i_n})$. Importantly, the sorted vector after the item has been given can be obtained by incrementing one entry. Specifically, the sorted vector will become $(w^t_{i_1}, \ldots, w^t_{{i_j}} + 1, \ldots  w^t_{i_n})$ where $i_j$ is the maximal $j$ such that the receiving agent had $w^t_{{i_j}}$ items. If $j = 1$ (the item was given to the agent with the fewest items), the inductive hypothesis clearly holds. We simply need to show that $w^t_{{i_{j - 1}}} \ge w^t_{{i_j}} - \ceil{ \log T \sqrt{T}}$. Importantly, the agent receiving the item $i$ must have $i \in S$, as defined on line 5. Thus $\{i_1, \ldots, i_j\} \subseteq S$ because each of these agents have currently received at most $i_j$ items. Furthermore, $w^t_{{i_{j - 1}}} \ge w^t_{{i_j}} - \ceil{ \log T \sqrt{T}}$, as otherwise $\{i_1, \ldots, i_{j - 1}\}$ satisfies the condition of line 5, and is strictly smaller in cardinality than the chosen $S$. Therefore, even after this addition $w^t_{{i_j}} + 1 - w^t_{{i_{j - 1}}} \le \ceil{ \log T \sqrt{T}}$.

To establish the general upper bound of $w^t_{i_n} \le (n - 1) \ceil{\log T \sqrt{T}}$, suppose for a contradiction there was a time step $t > T^{(1)}$ where $w_{i_n} > (n-1) \ceil{\log T \sqrt{T}}$. Then, by a straightforward induction over $j$, \[w_{i_{n + 1 - j}} > (n - j) \cdot \ceil{\log T \sqrt{T}},\] for all $1 \le j \le n$. Summing over all agents implies that phase 2 must last $> \frac{n(n - 1)}{2} \cdot \ceil{\log T \sqrt{T}}$. This is a contradiction.\qed

\subsection{Proof of \Cref{lem:halls}}
Fix values $a_1, \ldots, a_k$ and $b_1, \ldots, b_\ell$. Consider a bipartite graph with nodes $[k]$ on the left side and nodes $[\ell]$ on the right. We will have an edge $(i, j)$ precisely when $a_i \le b_j$. We would like to show that there is a matching in this graph of size at least $k - c$.

We first show that this is sufficient to imply the lemma. Let $M_L, M_R$ be the set of matched nodes and $U_L, U_R$ be the set of unmatched nodes on the left and right. We have that $\sum_{i \in M_L} a_i \le \sum_{i \in M_R} b_i$ by definition of the matching. Furthermore, $|U_L| \le c$, so $\sum_{i \in U_l} a_i \le c$. Putting this together, we have
\[
    \sum_i a_i = \sum_{i \in M_L} a_i + \sum_{i \in U_L} a_i \le \sum_{i \in M_R} b_i + c \le \sum_i b_i + c.
\]

To prove the existence of such a matching, it is sufficient so show that for all sets $S \subseteq [k]$, $|N(S)| \ge |S| - c$ where $N(S)$ is the neighborhood of $S$, i.e., all nodes in $[\ell]$ adjacent to at least one node in $S$~\cite{lovasz2009matching}. Fix such an $S$. Let $i \in \argmin_{i' \in S} a_{i'}$. Note that $S \subseteq \{i' \mid a_{i'} \ge a_i\}$ and $N(S) \supset \{i' \mid b_{i'} \ge a_i\}$, as all such nodes are adjacent to $a_i$. Thus, by the lemma condition $|S| \le |N(S)| + c$, as needed. \qed

\subsection{Proof of \Cref{lem:high-value-n}}

        Fix agents $i \ne j$. We will prove the statement is true for this pair of agents and then union bound over the at most $n(n - 1)$ pairs to yield the lemma statement. Without loss of generality, we will relabel $i$ as agent $1$ and $j$ as agent $2$. We will also assume $T$ is sufficiently large such that $T^{(1)} \ge T / 2$ and $\log T \sqrt{T} > 1$. The latter implies that $\ceil{\log T \sqrt{T}} \le 2 \log T \sqrt{T}$. 



        For each good $g$, let $I^g_i$ be the indicator variable denoting that agent $i$ has the highest quantile for item $g$. Given these random variables, agent $1$'s bundle at time $t \ge T^{(1)}$ takes on the form $$A^t_1 = \{g \in G^{welf} \mid I^g_1 = 1\} \cup \{g^1_1, \ldots g^1_k\}$$ for some value $k$ and agent $2$'s final bundle is $$A_2 = \{g \in G^{welf} \mid I^g_2 = 1\} \cup \{g^2_2, \ldots, g^2_\ell\}$$ for some value $\ell$. The conditions of the lemma statement hold precisely when $k \ge \ell + \ceil{\log T \sqrt{T}}$. Furthermore, note that by \Cref{lem:phase-2-items}, we only need to consider $k \le (n - 1) \cdot \ceil{\log T \sqrt{T}}$. 

        For each $q \in [0, 1]$ let \[X^0_q = \sum_{g \in G^{welf}} \mathbb{I}[Q^g_1 \ge q] \cdot I^g_1,\]
        i.e., the number of items agent 1 received during welfare maximization for which they had quantile at least $q$. Furthermore, for an integer $k$
        \[
            X^k_q = X^0_q + \sum_{j = 1}^k \mathbb{I}[Q^{g^1_j} \ge q],
        \]
        which counts the number of items $1$ has quantile $\ge q$ including the first $k$ items they receive in phase 2.

        Similarly, we will define $Y^k_q$ analogously for bundle 2. However, note that we still consider agent $1$'s bundle. More formally,

        \[
            Y^0_q = \sum_{g \in G^{welf}} \mathbb{I}[Q^g_1 \ge q] \cdot I^g_2 \text{ and } Y^k_q = Y^0_q + \sum_{j = 1}^k \mathbb{I}[Q^{g^2_j} \ge q],
        \]
Suppose $|A^t_1 \setminus G^{welf}| = k$ and $|A^t_2 \setminus G^{welf}| = \ell$. These random variables are useful because to show $\envy^t_{1, 2} \le c$, by \Cref{lem:halls}, it suffices to show that $\forall q \in [0, 1], X^k_q \ge Y^\ell_q + c$. Indeed, for any $g \in A^t_2$, let $q = \min_{g' \in A^t_1 \cup A^t_2: V^{g'}_1 \ge V^g_1} V^{g'}_1$. Then $X^k_q$ and $Y^\ell_q$ exactly count the number of items agent $1$ values at least as much as $g$ in $A^t_1$ and $A^t_2$, respectively.

To prove it is true for all time steps $t$ handled by the lemma condition, it suffices to show $\forall q \in [0, 1], X^k_q + c \ge Y^\ell_q$ for all $k \le (n - 1) \cdot \ceil{\log T \sqrt{T}}$ and $\ell \le k - \ceil{\log T \sqrt{T}}$. In fact, since both $X^k_q$ and $Y^\ell_q$ are nondecreasing in $\ell$ and $k$, it suffices to prove it for $k = \ell + \ceil{\log T \sqrt{T}}$. More concisely, our goal is to show
        \[
            \Pr[\forall q \in [0, 1], \forall \ell \in [(n - 2) \cdot \ceil{\log T \sqrt{T}}], \,  X_q^{\ell + \ceil{\log T \sqrt{T}}} + c \ge Y^\ell_q] \ge 1  - O(T^{-c/2}).
        \]
        Equivalently, we show
        \[
            \Pr[\exists q \in [0, 1], \exists \ell \in [(n - 2) \cdot \ceil{\log T \sqrt{T}}], \,  Y^\ell_q - X_q^{\ell + \ceil{\log T \sqrt{T}}}  > c ] \le  O(T^{-c/2}).
        \]
        To this end, we partition the unit interval into four subintervals. For each subinterval $[q_1, q_2]$, we show
        \begin{equation}\label{ineq:subinterval}
            \Pr[\exists q \in [q_1, q_2], \exists \ell \in [(n - 2) \cdot \ceil{\log T \sqrt{T}}], \,  Y^\ell_q - X_q^{\ell + \ceil{\log T \sqrt{T}}}  > c  ] \le  O(T^{-c/2}).
        \end{equation}
        and then apply a union bound over these four bounds to extend it to the entire interval. Each subinterval requires a different proof strategy. Note that by monotonicity of these variables showing $Y^\ell_{q_1} - X_{q_2}^k  > c $ implies $Y^{\ell'}_q - X_{q}^{k'}  > c $ for all $q \in [q_1, q_2]$, $\ell' \ge \ell$ and $k' \le k$.

        Before analyzing each subinterval, we introduce notation and derive bounds that will be useful throughout.


        For each item $g \in G^{welf}$, define
\[
    Z^g_i = Q^g_1 \cdot \mathbb{I}[I^g_i = 1] - \mathbb{I}[I^g_i = 0].
\]
That is, $Z^g_i$ equals $Q^g_1$ if $I^g_i = 1$ and $-1$ otherwise. These random variables give an alternate way to more directly compute $X^0_q$ and $Y^0_q$ as :
\[
    X^0_q = \sum_{g \in G^{welf}} \mathbb{I}[Z^g_1 \ge q] \text{ and } Y^0_q = \sum_{g \in G^{welf}} \mathbb{I}[Z^g_2 \ge q].
\]

Finally, it will be helpful to understand the distribution of each $Z^g_i$. Fix a good $g$. Define the CDFs $F^1$ and $F^2$ of $Z^g_1$ and $Z^g_2$, respectively.
For each $i$, by symmetry, $Z^g_i = -1$ with probability $1 - 1/n$, as each agent only receives an item during quantile maximization with probability $1/n$. With remaining probability, it matches the distribution of $Q^g_1 \mid I^g_i = 1$, the value of $Q^g_1$ conditioned on good $g$ going to agent $i$.
        
        For $i = 1$, the conditional $Q^g_1 \mid I^g_1 = 1$ follows a $\text{Beta}[n, 1]$ distribution, as it is the distribution of a uniform distribution conditional on it being the maximum of $n$ draws. In particular, the conditional CDF is $x^n$. Hence, for $x \in [0, 1]$,
\[
    F^1(x) = \frac{n- 1 + x^n}{n}.
\]
        
        For $i=2$, the distribution $Q^g_1 \mid I^g_2 = 1$ is equivalent a $\mathcal{U}[0, 1]$ conditioned on it  \emph{not} being the maximum of $n$ draws. Note that not being the maximum of $n$ draws occurs with with probability $(n - 1)/ n$. Thus, whatever this distribution $\mathcal{D}$ is, it must satisfy $$1/n \cdot \text{Beta}[n, 1] + (n - 1)/n \cdot \mathcal{D} = \mathcal{U}[0, 1]$$
        in the sense that an equivalent way of sampling from $\mathcal{U}[0, 1]$ is, with probability $1/n$ sample from $\text{Beta}[n, 1]$, and with remaining probability $(n-1)/n$ sample from $\mathcal{D}$. 
        Importantly, we can use this to solve for the CDF: $\frac{n}{n - 1} (x - \frac{x^n}{n})$. Hence, the unconditional CDF is \[F^2(x) = \frac{n - 1}{n} + \frac{1}{n - 1} \left( x - \frac{x^n}{n}\right).\]

        It will also be helpful to obtain more usable bounds on the probability $Z^g_i$ takes on values very close to $1$, i.e., bounds on $1 - F^i(1 - \varepsilon)$ for small values of $\varepsilon$. We have that,
\begin{align*}
    1 - F^1(1 - \varepsilon) &= \frac{1 - (1 - \varepsilon)^n}{n} \ge \frac{1 - \frac{1}{1 + \varepsilon \cdot n}}{n}\\
    &= \frac{\frac{\varepsilon \cdot n}{1 + \varepsilon \cdot n}}{n}\\
    &= \frac{\varepsilon}{1 + \varepsilon \cdot n}.
\end{align*}
Hence, for $\varepsilon \le 1/n$,
\[
    1 - F^1(1 - \varepsilon) \ge \frac{\varepsilon}{2}.
\]
For $F^2$, a necessary condition for $Z^g_2 \ge 1 - \varepsilon$ is $Q^g_1 \ge 1 - \varepsilon$ and $I^g_2 = 1$. The latter implies that $Q^g_2 \ge Q^g_1$, and thus, $Q^g_2 \ge 1 - \varepsilon$ as well. The probability that both $Q^g_1 \ge 1 - \varepsilon$ and $Q^g_2 \ge 1- \varepsilon$ is $\varepsilon^2$. Hence,
\[
    1 - F^2(1 - \varepsilon) \le \varepsilon^2.
\]

Finally, we will upper bound the probability that $Z^g_1$ is small (but not $-1$). More specifically, $\Pr[Z^g_1 \in [0, \varepsilon]] \le \varepsilon^n$. Indeed, a necessary condition for this to occur is that $Q^g_1 \le \varepsilon$ and $Q^g_i \le Q^g_1$ for all $i \ne 1$. This implies $Q^g_i \le \varepsilon$ must hold for all $i$. This only occurs with probability $\varepsilon^n$. For our purposes, it will be sufficient to use the weaker bound
\begin{equation}\label{ineq:small-eps}
    \Pr[Z^g_1 \in [0, \varepsilon]] \le \varepsilon^2.
\end{equation}

With these facts in hand, we now analyze each subinterval.

\paragraph{Part 1: $\left[0, \frac{8n^2\log T}{\sqrt{T}}\right]$.} Let $q = \frac{8n^2\log T}{\sqrt{T}}$, and fix an arbitrary $\ell \in [(n - 2) \cdot \ceil{\log T \sqrt{T}}]$. We will show that $\Pr[Y^\ell_q - X_q^{\ell + \ceil{\log T \sqrt{T}}}  > c] \le T^{\Omega(\log(T))}$. Union bounding over the $(n - 2) \cot \ceil{\log T \sqrt{T}} + 1$ choices of $\ell$ yields the desired result.

For each good $g \in G^{welf}$, let $W^g \in \{-1, 0, 1\}$ be a random variable such that $W^g = -1$ if $Z^g_1 \ge q$ ($g$ counts toward $X_q^{\ell + \ceil{\log T \sqrt{T}}}$), $1$ if $Z^g_i \ge 0$ ($g$ counts toward $Y^\ell_0$) and $0$ otherwise. Furthermore, for $g \in \{g^1_1, \ldots, g^1_{\ell + \ceil{\log T \sqrt{T}}}\}$, let $W^g \in \{-1, 0\}$ be such that $W^g = -1$ if $Q^g_1 \ge q$ ($g$ counts toward $X_q^{\ell + \ceil{\log T \sqrt{T}}}$). Finally, for $g \in \{g^2_1, \ldots, g^1_{\ell}\}$, let $W^g = 1$ denoting that $g$ counts toward $Y_q^{\ell}$. Importantly, we have $\sum_g W^g = X_q^{\ell + \ceil{\log T \sqrt{T}}} - Y^\ell_q$. We will prove that $\sum_g W^g \le 0$ with high probability using Hoeffding's inequality. 

Let us now consider $\E[\sum_g W^g]$. For $g \in G^1$, $\E[W^g] = - (1 - q)$, because it is $-1$ as long as $Q^g_1 \ge q$. For $g \in G^2$, $W^g = 1$ deterministically, so $\E[W^g] = 1$ as well. For the remaining $g \in G^{welf}$, note that $g$ is given to each of agents $1$ and $2$ with probability $1/n$ each. However, with probability at most $q^2$, $g$ is given to agent $1$ with $Q^g_1 \le q$ by \eqref{ineq:small-eps}. Hence, $\mathbb{E}[W^g] \le q^2$. Furthermore, there are at most $T^{(1)} \le T$ such goods $g$

Putting these together, we have that
\begin{align*}
    \E\left[\sum_g W^g\right]
    &\le -(\ell + \ceil{\log T \sqrt{T}}) \cdot (1 - q) + \ell + T \cdot q^2\\
    &= -\ceil{\log T \sqrt{T}}(1 - q) + \ell \cdot q + T q^2.\\
    &\le  -\ceil{\log T \sqrt{T}}(1 - q) + (n - 2) \cdot  \ceil{\log T \sqrt{T}} \cdot q + T q^2\\
    &\le -\ceil{\log T \sqrt{T}}(1 - (n - 2) \cdot q) + T q^2\\ 
    &= -\ceil{\log T \sqrt{T}}(1 - (n - 2) \cdot q) + 8n^2 \log T \sqrt{T} \cdot q\\
    &\le -\ceil{\log T \sqrt{T}}(1 - (n - 2) \cdot q) + 8n^2 \ceil{\log T \sqrt{T}} \cdot q\\
    &\le -\ceil{\log T \sqrt{T}}\left(1 - (8n^2 + n - 2) \cdot \frac{8n^2 \log T}{\sqrt{T}}\right)\\
    &\le -\log T \sqrt{T} / 2
\end{align*}
where the last transition follows under the assumption that $T$ is sufficiently large such that $\sqrt{T} / \log T \ge (8n^2 + n - 2) \cdot 8n^2  \cdot 2$. 

Now, $\sum_g W^g$ is the sum of $\le T$ independent random variables bounded by $[-1, 1]$. Thus, Hoeffding's inequality ensures that a deviation of $\log T \sqrt{T} / 2$ occurs with probability at most \[\exp\left(-\frac{2(\log T \sqrt{T} / 2)^2}{4 \cdot T}\right) = \exp(\log^2(T)/8) = T^{\log T / 8},\]
as needed.
        
\paragraph{Part 2: $\left[\frac{8n^2\log T}{\sqrt{T}}, 1 - \frac{8n^2\log T}{\sqrt{T}} \right]$.} For this interval, we will use the DKW inequality to show that, with high probability, the collection of $\{Z^g_1\}_{g \in G^{welf}}$ and $\{Z^g_2\}_{g \in G^{welf}}$ approximately match their true distributions. Let $\hat{F}^1$ and $\hat{F}^2$ be the empirical CDFs of $\{Z^g_1\}_{g \in G^{welf}}$ and $\{Z^g_2\}_{g \in G^{welf}}$. Let $\varepsilon = \frac{n\log T}{\sqrt{T}}$. The DKW inequality~\cite{dvoretzky1956asymptotic} states that
        \[
            \Pr[\sup_x|\hat{F}^i(x) - F^i(x)| > \varepsilon] \le 2\exp(-2T^{(1)} \varepsilon^2).
        \]
        Since $T^{(1)} \ge T/2$, this expands to
        \[
            2\exp(-2T^{(1)} \varepsilon^2) \le 2\exp(-T \cdot \log^2 T \cdot n^2 / T) = 2\exp(-\log^2 T n^2) = 2T^{-n^2 \log T} \in O(T^{-c/2}).
        \]
        Furthermore, we claim that conditioned on this event holding for both $i \in \{1, 2\}$, along with the assumption on the sampled quantiles that each $Q^g_i$ is distinct, \eqref{ineq:subinterval} holds. Indeed, fix a $q \in \left[\frac{8n^2\log T}{\sqrt{T}}, 1 - \frac{8n^2\log T}{\sqrt{T}} \right]$. We have that
        \begin{align*}
            X^0_q - Y^{(n - 2) \ceil{\log T \sqrt{T}}}_q
            &\ge X^0_q - Y^0_q - (n - 2) \ceil{\log T \sqrt{T}}\\
            &\ge  T^{(1)} \cdot (1 - \hat{F}^1(q)) - T^{(1)} \cdot ( 1 - \hat{F}^2(q)) -1  - (n - 2) \ceil{\log T \sqrt{T}}\\
            &= T^{(1)} \cdot (\hat{F}^2(q) - \hat{F}^1(q)) - 1 - (n - 2) \ceil{\log T \sqrt{T}}\\
            &> T^{(1)} \cdot \left(F^2(q) - F^1(q) - 2 \varepsilon\right) - 1 - (n - 2) \ceil{\log T \sqrt{T}}\\
            &= T^{(1)} \cdot \left(\frac{q}{n - 1} - \frac{q^n}{n(n - 1)} - \frac{q^n}{n} - 2 \varepsilon\right) - 1 - (n - 2) \ceil{\log T \sqrt{T}}\\
            &= T^{(1)} \cdot \left(\frac{q - q^n}{n - 1} - 2 \varepsilon\right) - 1 - (n - 2) \ceil{\log T \sqrt{T}}\\
            &=  T^{(1)} \cdot \left(\frac{q(1 - q^{n-1})}{n - 1} - 2 \varepsilon\right) - 1 - (n - 2) \ceil{\log T \sqrt{T}}\\
            &\ge  T^{(1)} \cdot \left(\frac{q(1 - q)}{n - 1} - 2 \varepsilon\right) - 1 - (n - 2) \ceil{\log T \sqrt{T}}\\
            &\ge  T^{(1)} \cdot \left(\frac{\min(q, 1 - q)}{2(n - 1)} - 2 \varepsilon\right) - 1 - (n - 2) \ceil{\log T \sqrt{T}}\\
            &\ge  T / 2 \cdot \left(\frac{8n^2 \log T}{2(n - 1) \cdot \sqrt{T}} - \frac{2n \log T}{\sqrt{T}} \right) - 1 - (n - 2)\ceil{\log T \sqrt{T}}\\
            &\ge  T / 2 \cdot \left(\frac{4n \log T}{\sqrt{T}} - \frac{2n \log T}{\sqrt{T}} \right) - 1 - (n - 2)\ceil{\log T \sqrt{T}}\\
            &\ge  2n \log T \sqrt{T} - 1 - (n - 2)\ceil{\log T \sqrt{T}}\\
            &\ge -1 \ge -c,
        \end{align*}
        as needed.

    \paragraph{Part 3: $\left[1 - \frac{ 8n^2 \log T}{\sqrt{T}}, 1 - \frac{400n^4 \log^2T}{T} \right]$} Let $q_1 = 1 - \frac{ 8n^2 \log T}{\sqrt{T}}$ and $q_2 =  1 - \frac{400n^2 \log^2T}{T}$. We will show that with high probability $Y^{(n - 2) \ceil{\log T \sqrt{T}}}_{q_1} \le 90n^2\log^2 T$, and, similarly, $X^0_{q_2} \ge 90n^2\log^2 T$ with high probability. Together, these imply that $X^0_{q_1} - Y^{(n - 2) \ceil{\log T \sqrt{T}}}_{q_2} < c$ occurs with probability $O(T^{-c/2})$.

    To that end, let us first consider $Y^{(n - 2) \ceil{\log T \sqrt{T}}}_{q_1}$. We have that
\begin{align*}
    \mathbb{E}[Y^{(n - 2) \ceil{\log T \sqrt{T}}}_{q_2}] 
    &\le T^{(1)} \cdot (1 - q_1)^2 + (n - 2) \ceil{\log T \sqrt{T}} \cdot (1 - q_1)\\
    &\le 64n^4 \log^2 T + (n - 2) \cdot 8 n^2 \log T \ceil{\log T \sqrt{T}}\\
    &\le 64n^2 \log^4 T + 16n^4 \log^2 T \sqrt{T} = 80n^4 \log^2 T.
\end{align*}
Thus, a standard Chernoff bound implies that \[\Pr[Y^{(n - 2) \ceil{\log T \sqrt{T}}}_{q_2} \ge 90n^4 \log^2 T] \le \exp(- 80n^4\log^2 T \cdot (1/8)^2 / (2 + 1/8)) \le T^{n^4 \log T / 2} \in O(T^{-c/2}).\]
Next, let us consider $X^0_{q_2}$. We have that
\begin{align*}
    \mathbb{E}[X^0_{q_2}]
    &\ge T^{(1)} \cdot  (1 -q_2)/ 2\\
    &\ge T \cdot (1 - q_2) / 4\\
    & \ge 100n^2 \log^2 T. 
\end{align*}
Hence, a standard Chernoff bound implies that
\[
    \Pr[X^0_{q_2} \le 90 n^4 \log^2T] \le \exp(- 100n^4\log^2 T(1/10^2) / 2) = T^{n^4 \log T / 2} \in (T^{-c/2}).
\]

\paragraph{Part 4: $\left[1 - \frac{400n^2 \log^2T}{T}, 1 \right]$}
Let $q = 1 - \frac{400n^2 \log^2T}{T}$. We will show that with high probability $Y^{(n - 2) \ceil{\log T \sqrt{T}}}_{q} \le c$. Note that $Y^{(n - 2) \ceil{\log T \sqrt{T}}}_{q}$ is integer valued, so it is sufficient to upper bound the probability it is above $c + 1$. To that end,
\begin{align*}
    \mathbb{E}[Y^{(n - 2) \ceil{\log T \sqrt{T}}}_{q}] 
    &\le T^{(1)} \cdot (1 - q_1)^2 + L \cdot (1 - q_1)\\
    &\le 160000n^8\log^4 T / T +   400 n^4 \log T \ceil{\log T \sqrt{T}} / T\\
    &\le 160000n^8\log^4 T / T +   800 n^4 \log^2 T / \sqrt{T}\\
\end{align*}
Note that this value is $O(\log^2 T / \sqrt{T})$. We will use the Chernoff bound which states that
\[
		\Pr[W \ge (1 + \delta) \mu] \le \left(\frac{e^\delta}{(1 + \delta)^{1 + \delta}} \right)^\mu \leq  \left(\frac{e}{(1 + \delta)} \right)^{(1 + \delta)\mu}
\]
for a random variable $W$ with mean $\mu$. In this case, if we set $\delta$ such that $1 + \delta = (c + 1)/\mathbb{E}[Y^{(n - 2) \ceil{\log T \sqrt{T}}}_{q}]$, note that $\delta \in \Omega(\sqrt{T} / \log^2 T)$. Thus this bound implies an overall probability bound of \[O\left(\left(\frac{\log^2 T}{\sqrt{T}}\right)^{c + 1}\right) \in O(T^{-c/2}),\] as needed. \qed

\end{document}